\documentclass[11pt]{article}
\usepackage[utf8]{inputenc}
\usepackage[T1]{fontenc}
\usepackage[top=1in, bottom=1in, left=1in, right=1in, letterpaper]{geometry}

\usepackage{amsmath, amssymb, amsfonts, bm, cite}
\usepackage{amsthm}
\usepackage{nicefrac}
\usepackage{authblk}
\usepackage{mathtools}
\usepackage{thmtools} 
\usepackage{thm-restate}

\usepackage{enumitem}
\usepackage{comment}
\usepackage{xspace}
\usepackage{ifthen}
\usepackage{boxedminipage}

\usepackage[ruled, vlined, linesnumbered]{algorithm2e}
\SetKwProg{Init}{initialize}{}{}

\SetKwFunction{pEnvyFreeCycle}{EnvyFreeCycle}
\SetKwFunction{pLocalSearch}{LocalSearch}
\SetKwFunction{pFairOrEfficient}{MakeFairOrEfficient}
\usepackage{color}
\usepackage{empheq}
\usepackage[colorlinks=true,allcolors=blue,bookmarksnumbered=true,hyperfootnotes=true]{hyperref}

\usepackage{tikz}
\usetikzlibrary{shapes,arrows}
\usetikzlibrary{intersections}

\newtheorem{theorem}{Theorem}[section]
\newtheorem{lemma}[theorem]{Lemma}
\newtheorem{claim}[theorem]{Claim}
\newtheorem*{claim*}{Claim}
\newtheorem{proposition}[theorem]{Proposition}

\newtheorem{corollary}[theorem]{Corollary}

\newtheorem{remark}[theorem]{Remark}

\theoremstyle{definition}
\newtheorem{definition}[theorem]{Definition}

\newcommand{\RR}{{\mathbb R}}

\newcommand{\cR}{\mathcal{R}}
\newcommand{\cS}{\mathcal{S}}
\newcommand{\cT}{\mathcal{T}}
\newcommand{\cL}{\mathcal{L}}
\newcommand{\cBP}{\mathcal{K}}

\newcommand{\aA}{A}
\newcommand{\aE}{E}
\newcommand{\aAn}{\bar A}

\newcommand{\aG}{M}
\newcommand{\aJ}{J}
\newcommand{\ee}{\mathrm{e}}

\newcommand{\eps}{\varepsilon}
\newcommand{\nfrac}{\nicefrac}

\DeclareMathOperator*{\NSW}{NSW} 
\DeclareMathOperator{\OPT}{OPT} 
\DeclareMathOperator*{\argmax}{argmax}

\newcommand\blfootnote[1]{%
  \begingroup
  \renewcommand\thefootnote{}\footnote{#1}%
  \addtocounter{footnote}{-1}%
  \endgroup
}

\title{Approximating Nash Social Welfare\\ by Matching and Local Search\blfootnote{This article represents the unified 
journal version of three conference papers: the STOC 2021 paper \cite{GHV20} by JG, EH, and LAV; the FOCS 2022 paper \cite{LiV22} by WL and JV; and the STOC 2023 paper \cite{GHLVV23} by all five authors. 
JV and WL were supported by NSF Award 2127781. LAV received funding from the European Research Council (ERC) under the European Union's Horizon 2020 research and innovation programme (grant agreement no. ScaleOpt--757481). EH was supported by SNSF Grant 200021 200731/1. JG was supported by NSF Grants CCF-1942321 and CCF-2334461.}}

\date{{\tt jugal@illinois.edu, edinehusic@gmail.com, wzli@uni-bonn.de, lvegh@uni-bonn.de, jvondrak@stanford.edu}}

\author[1]{Jugal Garg}
\author[2]{Edin Husi\' c}
\author[3,5]{Wenzheng Li} 
\author[4,5]{L{\'{a}}szl{\'{o}} A. V{\'{e}}gh}
\author[3]{Jan Vondr\'{a}k}
\affil[1]{University of Illinois at Urbana-Champaign, IL}
\affil[2]{IDSIA, USI-SUPSI, Switzerland}
\affil[3]{Stanford University, CA}
\affil[4]{London School of Economics and Political Science, UK} 
\affil[5]{University of Bonn, Germany} 

\begin{document}
\maketitle 

\begin{abstract}
For any $\eps>0$, we give a simple, deterministic $(4+\eps)$-approximation algorithm for the Nash social welfare (NSW) problem under submodular valuations.  We also consider the asymmetric variant of the problem, where the objective is to maximize the weighted geometric mean of agents' valuations, and give an $\ee (\omega + 2 + \eps)$-approximation if the ratio between the largest weight and the average weight is at most $\omega$. 

We also show that the $\nfrac12$-EFX envy-freeness property can be attained simultaneously with a constant-factor approximation. More precisely, we can find an allocation in polynomial time that is both $\nfrac12$-EFX and an $(8+\eps)$-approximation to the symmetric NSW problem under submodular valuations. 
\end{abstract}

\section{Introduction}
We consider the problem of allocating a set $\aG$ of $m$ indivisible items among a set $\aA$ of $n$ agents, where each agent $i\in\aA$ has a valuation function $v_i: 2^{\aG}\to \RR_{\ge0}$ and a weight (entitlement) $w_i>0$ such that $\sum_{i\in\aA} w_i = 1$. The Nash social welfare (NSW) problem asks for an allocation ${\mathcal{S}}=(S_i)_{i\in \aA}$ that maximizes the weighted geometric mean of the agents' valuations,
\[ \NSW({\mathcal{S}}) = \prod_{i\in \aA} 
\left(v_i(S_i)\right)^{w_i}.\]
We refer to the special case when all agents have equal weight (i.e., $w_i=1/n$) as the \emph{symmetric} NSW problem, and call the general case the \emph{asymmetric} NSW problem.  Throughout, we let $w_{\max} \coloneqq \max_{i\in \aA} w_i$. For $\alpha>1$, an \emph{$\alpha$-approximate solution} to the NSW problem is an allocation $\cS$ with $\NSW(\cS)\ge \mathrm{OPT}/\alpha$, where $\mathrm{OPT}$ denotes the optimum value of the NSW-maximization problem.

Allocating resources among agents in a fair and efficient manner is a fundamental problem in computer science, economics, and social choice theory; we refer the reader to the monographs~\cite{Barbanel,Brams1996,BrandtCELP16,Moulin2004,Robertson1998,R16,Young1995} on the background. A common measure of efficiency is \emph{utilitarian social welfare}, i.e., the sum of the utilities  $\sum_{i\in \aA} v_i(S_i)$ for an allocation $(S_i)_{i\in \aA}$.
In contrast, fairness is often measured by \emph{max-min fairness}, i.e., $\min_{i\in\aA} v_i(S_i)$; maximizing this objective is also known as the \emph{Santa Claus problem} \cite{bansal2006santa}. 

Symmetric NSW provides a balanced tradeoff between the often conflicting requirements of fairness and efficiency.  It has been introduced independently in a variety of contexts. It is a discrete analogue of the Nash bargaining game~\cite{Kaneko1979,nash1950bargaining}; it corresponds to the notion of competitive equilibrium with equal incomes in economics~\cite{varian1973equity}; and arises as a proportional fairness notion in networking~\cite{kelly1997charging}. The more general asymmetric objective has also been well-studied since the seventies \cite{harsanyi1972generalized,kalai1977nonsymmetric}. It has found applications in different areas, such as bargaining theory~\cite{chae2010bargaining,Laruelle2007}, water resource allocation~\cite{degefu2016water,houba2013asymmetric}, and climate agreements~\cite{yu2017nash}.

A distinctive feature of the NSW problem is invariance under scaling of the valuation functions $v_i$ by independent factors $\lambda_i$, i.e., each agent can express their preference in a ``different currency'' without changing the optimization problem (see~\cite{Moulin2004} for additional characteristics).

\paragraph{Submodular and subadditive valuation functions}
A set function $v:\, 2^{\aG} \to \RR$  is \emph{monotone} if $v(S)\le v(T)$ whenever $S\subseteq T$. A monotone set function with $v(\emptyset) = 0$ is also called a \emph{valuation function} or simply {\em valuation}.
The function $v:\, 2^{\aG} \to \RR$ is \emph{submodular} if
\[ v(S) + v(T) \ge v(S\cap T) + v(S\cup T)\, \quad  \forall S, T \subseteq   \aG\, ,\]
and \emph{subadditive} if  
\[ v(S) + v(T) \ge v(S\cup T)\, \quad  \forall S, T \subseteq   \aG\, .\]
We assume the valuation functions are given by value oracles that return $v(S)$ for any $S\subseteq \aG$ in $O(1)$ time. %

\paragraph{$\nfrac12$-EFX allocations} %
Envy-freeness is one of the most desirable fair division concepts. However, for indivisible goods, exact envy-freeness is generally unattainable: if there is only one desirable item, all but one agent will necessarily be envious. This limitation has motivated a rich line of work introducing and analyzing various relaxations. Among these, envy-freeness
up to any item (EFX) is the strongest and widely regarded as the most compelling fairness notion in the discrete setting with equal entitlements~\cite{caragiannis2019unreasonable}.  An allocation $\cS=(S_i)_{i\in \aA}$ is said to be EFX if
\[v_i(S_i) \ge v_i(S_k \setminus \{j\}), \enspace\forall i, k \in \aA\,,  \forall j\in S_k\ .\]
That is, no agent envies another agent's bundle after the removal of any single item from that agent's bundle. It is not known whether EFX allocations always exist or not, even in the additive case with more than three agents, and it is regarded as ``fair division's biggest open question''~\cite{Procaccia20}. This motivated the study of its relaxation $\alpha$-EFX for an $\alpha \in (0, 1)$, where an allocation $\cS$ is said to be $\alpha$-EFX if 
\[v_i(S_i) \ge \alpha\cdot v_i(S_k \setminus \{j\}), \enspace\ \forall i, k \in \aA\, ,\forall j\in S_k\, .\]
The best-known $\alpha$, for which the existence is known, is $\nfrac12$ for submodular valuations, albeit with the efficiency guarantee of $O(n)$-approximation to the symmetric NSW problem~\cite{PlautR18, chaudhury2021fair}.

\paragraph{Our contributions}
Our main theorem on NSW is the following.
\begin{theorem}\label{thm:main}
For any $\eps>0$,
there is a deterministic polynomial-time 
$(nw_{\max}+2+\eps) \ee$-approxi\-mation algorithm for the asymmetric Nash social welfare problem with submodular valuations. 
For symmetric instances, the algorithm returns a $(4+\varepsilon)$-approximation. The number of arithmetic operations and value oracle calls is polynomial in $n$, $m$, and $1/\eps$.
\end{theorem}
 Algorithm~\ref{alg:NSW-template} in Section~\ref{sec:nsw-alg} presents the algorithm asserted in the theorem.
Note that $n w_{\max}$ is the ratio between the maximum weight $w_{\max}$ and the average weight $1/n$. In the symmetric case, when all weights are $w_i = {1}/{n}$, this bound gives $(3+\eps)\ee < 8.2$. In this case, we can improve the analysis to obtain a $(4+\eps)$-approximation 
algorithm. Very recently, Bei, Feng, Hu, Li, and Zhang \cite{bei2025nashsocialwelfaresubmodular}  tightened the analysis of the algorithm in Theorem~\ref{thm:main}, improving the factor  to $3.56+\varepsilon$.

\medskip

As our second result, we show that a $\nfrac12$-EFX allocation with high NSW value exists and can also be efficiently found. We give a general reduction for subadditive valuations.
An allocation $\cS=(S_i)_{i\in \aA}$ is called \emph{complete} if every item is included, i.e., $\bigcup_{i\in\aA}S_i=\aG$; otherwise, it is called \emph{partial}. 
 In the context of  $\nfrac12$-EFX allocations, $\NSW(\cS)$ will always refer to the NSW value of allocation $\cS$ in the symmetric case ($w_i=1/n$ for all $i\in \aA$).

\begin{restatable}{theorem}{efxhalf}\label{thm:efx2}
There is a deterministic strongly polynomial-time algorithm that given a symmetric NSW instance with subadditive valuations and given a (complete or partial) allocation $\cS$ of the items, it  returns a complete allocation $\cT$ that is $\nfrac12$-EFX and $\NSW(\cT)\ge \NSW(\cS)/2$. 
\end{restatable}
The above algorithm is strongly polynomial in the value oracle model: the number of basic arithmetic operations and oracle calls is polynomially bounded in $n$ and $m$.  Together with Theorem~\ref{thm:main}, we obtain the following corollary.
\begin{corollary}
For any $\eps>0$,
there is a deterministic polynomial algorithm that returns a $\nfrac12$-EFX complete allocation that is {$(8+\varepsilon)$}-approximation to the symmetric NSW problem under submodular valuations. The number of arithmetic operations and value oracle calls is polynomial in $n$, $m$, and $1/\eps$.
\end{corollary}

\subsection{Related work}
\paragraph{Prior work on approximating NSW} 
Let us first consider  \emph{additive valuations}, i.e., when $v_i(S) = \sum_{j \in S} v_{ij}$ for nonnegative values $v_{ij}$. 
Maximizing symmetric NSW is NP-hard already in the case of two agents with identical additive valuations, by a reduction from the Subset Sum problem. It is NP-hard to approximate within a factor better than $1.069$ for additive valuations~\cite{garg2017satiation}, and better than $1.5819$ for submodular valuations~\cite{GargKK20}. 

On the positive side, a number of remarkably different constant-factor approximations are known for additive valuations. The first such algorithm with the factor of $2\cdot e^{1/e} \approx 2.889$ was given by Cole and Gkatzelis~\cite{cole2015approximating} using a continuous relaxation based on a particular market equilibrium concept. Later,~\cite{cole2017convex} improved the analysis of this algorithm to achieve the factor of 2. Anari, Oveis Gharan, Saberi, and Singh~\cite{anari2017nash} used a convex relaxation that relies on properties of real stable polynomials. The current best factor is $\ee^{1/\ee}+\eps \simeq 1.45$ by Barman, Krishnamurthy, and Vaish~\cite{barman2018finding}; the algorithm uses a different market equilibrium based approach and also guarantees approximate EF1 fairness, a weaker version of the EFX property. 

For the general class of subadditive valuations~\cite{BBKS20,chaudhury2021fair}, $O(n)$-approximations are known. This is the best one can hope for in the value oracle model~\cite{BBKS20}, for the same reasons that this is impossible to improve for the utilitarian social welfare problem~\cite{DNS10}. Sublinear $O(n^{53/54})$ approximation was designed for XOS valuations if we are given access to both demand and XOS oracles~\cite{barman2021sublinear}. Recall that all submodular valuations are XOS, and all XOS valuations are subadditive. 

Constant-factor approximations were also obtained beyond additive valuation functions: capped-additive~\cite{garg2018approximating}, separable piecewise-linear concave (SPLC)~\cite{anari2018nash}, and their common generalization, capped-SPLC~\cite{ChaudhuryCGGHM22} valuations; the approximation factor for capped-SPLC valuations  matches the $\ee^{1/\ee}+\varepsilon$ factor for additive valuations. All these valuations are special classes of submodular. Subsequently, Li and Vondr\'ak~\cite{li2021estimating} designed an algorithm that estimates the optimal value within a factor of $\frac{\ee^3}{(\ee-1)^2} \simeq 6.8$ for a broad class of submodular valuations, such as coverage and summations of matroid rank functions, by extending the techniques of~\cite{anari2017nash} using real stable polynomials. However, this algorithm only estimates the optimum value but does not find a corresponding allocation in polynomial time.

In \cite{GHV20}, Garg, Husi\'c, and V\'egh developed a constant-factor approximation for a subclass of submodular valuations called \emph{Rado-valuations}. These include  weighted matroid rank functions and many others that can be obtained using operations such as induction by network and contractions. An important example outside this class is the coverage valuation. They attained an approximation ratio 772 for the symmetric case and {$772(w_{\max}/w_{\min})^3$} for the asymmetric case. 
Subsequently, Li and Vondr\'ak \cite{LiV22} obtained a randomized 
380-approximation for symmetric NSW under submodular valuations by extending the 
the approach of \cite{GHV20}.

We significantly improve and simplify the approach used in \cite{GHV20} and \cite{LiV22}; we give a comparison to these works in Section~\ref{sec:overview-GHVLiV}.

\paragraph{Subsequent work on approximating NSW}
Subsequent to the conference version of this work \cite{GHLVV23}, Dobzinski, Li, Rubinstein, and Vondr\'ak \cite{dobzinski2023constant} obtained a constant factor approximation for symmetric NSW for subadditive valuations, assuming a demand oracle. They extend the approach of \cite{GHV20} and \cite{LiV22}, using also a configuration LP rounding. 

Further, Brown, Laddha, Pittu, and Singh \cite{brown2024approximation} obtained an approximation factor of $5\exp(2\log n + 2\sum_{i\in\aA} w_i\log w_i)$ for additive valuations. The term in $\exp(.)$ is twice the Bregman divergence between the distribution $(w_i)_{i\in \aA}$ and the uniform distribution.
As pointed out by Bento Natura~\cite{bento}, a simple change in our analysis leads to the improved bound $\exp(\log(3n)+\sum_{i\in\aA} w_i\log w_i)\cdot (\ee+\eps)$ in Theorem~\ref{thm:main} for general submodular valuations; this is explained in Remark~\ref{rem:bento}.
Feng and Li~\cite{fengli} obtained an $(\ee^{1/\ee}+\varepsilon)$-approximation for general asymmetric NSW with additive valuations using a novel rounding of the configuration LP. Soon after, their approach was generalized by Feng, Hu, Li, and Zhang \cite{feng2024constant} to get a $(233+\varepsilon)$-approximation algorithm for the asymmetric NSW with general submodular valuations. Most recently,  Bei, Feng, Hu, Li, and Zhang \cite{bei2025nashsocialwelfaresubmodular} further improved this ratio to $5.18$, as well as tightened the analysis Theorem~\ref{thm:main} as mentioned above.

\paragraph{Prior work on EFX and related notions} The existence of EFX allocations has not been settled despite significant efforts~\cite{caragiannis2019unreasonable, PlautR18, ChaudhuryGM20, Procaccia20}. This problem is open for more than two agents with general monotone valuations (including submodular), and for more than three agents with additive valuations. This necessitated the study of its relaxations, $\alpha$-EFX, for $\alpha\in(0,1)$, and of partial EFX allocations. For the notion of $\alpha$-EFX, the best-known $\alpha$ is $0.618$ for additive~\cite{amanatidis2020mnwefx} and $0.5$ for subadditive valuations (including submodular)~\cite{PlautR18}.

For the notion of partial EFX allocations, the existence is known for general monotone valuations if we do not allocate at most $n-2$ items~\cite{CKMS21,Mahara21,BergerCFF}, albeit without any efficiency guarantees. For additive valuations, although $n-2$ is still the best bound known, there exist partial EFX allocations with a 2-approximation to the NSW problem~\cite{CaragiannisGH19}.

A well-studied weaker notion is envy-freeness up to one item (EF1), where no agent envies another agent after the removal of \emph{some} item from the envied agent's bundle. EF1 allocations are known to exist for general monotone valuations and can also be computed in polynomial time~\cite {LiptonMMS04}. However, an EF1 allocation alone is not desirable because it might be highly inefficient in terms of any welfare objective.  For additive valuations, the allocations maximizing NSW are EF1~\cite{caragiannis2019unreasonable}. Although the NSW problem is APX-hard~\cite{lee2017apx}, there exists a pseduopolynomial time algorithm to find an allocation that is {approximately} EF1 and 1.45-approximation to the NSW problem under additive valuations~\cite{barman2018finding}. For capped-SPLC valuations,~\cite{ChaudhuryCGGHM22} shows the existence of an allocation that is $\nfrac12$-EF1 and $1.45$-approximation to the NSW problem. 

Subsequent to the conference version of this work, Feldman, Mauras, and Ponitka~\cite{FeldmanMP23} improved Theorem~\ref{thm:efx2} to show the existence of an allocation $\cT$ that is $\nfrac12$-EFX and $\NSW(\cT) \ge \nfrac23 \NSW(\cS)$ for a given allocation $\cS$. For subadditive valuations,~Barman and Suzuki~\cite{BarmanS24} recently showed the existence of a complete allocation that is EF1 and provides a 2-approximation to the NSW problem, as well as of a partial allocation that is  EFX and also provides a 2-approximation to the NSW problem.

\paragraph{Notation}
We use $\log(x)$ for the natural logarithm throughout. For set $S\subseteq \aG$ and $j\in \aG$, we use $S+j$ to denote $S\cup\{j\}$ and $S-j$ for $S\setminus \{j\}$ and we write $v(j)$ for $v(\{j\})$. For a vector $p\in \RR^\aG$ and $S\subseteq \aG$, we denote $p(S)=\sum_{i\in S}p_i$.

By a \emph{matching} from $\aA$ to $\aG$ we mean a mapping $\tau:\, \aA\to \aG\cup\{\bot\}$ where $\tau(i)\neq \tau(j)$ if $\tau(i)\neq \bot$; $\bot$ is a special symbol representing unmatched agents. For $X\subseteq A$, a  matching is \emph{$X$-perfect}, if $\tau(i)\neq\bot$ for every $i\in X$.
An $A$-perfect matching is simply a perfect matching when every node on both sides is matched, which is possible when $|A| = |\aG|$.

Recall that the definition of valuation functions requires $v(\emptyset)=0$. In our algorithm, we will use the concept of endowed valuation functions. The set function $v\, :\, 2^\aG\to \mathbb{R}$ is an \emph{endowed valuation function}, if it is monotone, and  $2v(\emptyset)\ge\max_{g\in \aG} v(g)$, i.e.,  the initial value is already at least half the value of any singleton.

\paragraph{Overview of the paper}
Section~\ref{sec:nsw-alg} describes the NSW approximation algorithm. A comparison of our technique with previous approaches is given in Section~\ref{sec:overview-GHVLiV}.  
The exposition of the fairness guarantee is given in Section~\ref{sec:fairness}.
We conclude the paper with some remarks in Section~\ref{sec:conclusion}.

\section{Approximation algorithm for Nash social welfare}\label{sec:nsw-alg}
The NSW algorithm asserted in Theorem~\ref{thm:main} is shown in
Algorithm~\ref{alg:NSW-template}.
The subroutine \pLocalSearch will be described in Algorithm~\ref{alg:localSearch}.  Throughout, we assume $0<\varepsilon\le 1$; we can always replace $\varepsilon$ by $\min\{\varepsilon,1\}$.
The algorithm proceeds in the following three phases.

\begin{algorithm}[htbp]
\DontPrintSemicolon
\KwIn{Valuations $(v_i)_{i\in \aA}$ over $\aG$, weights $w\in \mathbb{R}_{>0}^\aA$ such that $\sum_{i\in\aA} w_i = 1$, and $\varepsilon>0$.}
\KwOut{Allocation $\cS=(S_i)_{i\in \aA}$.}
Find an $A$-perfect matching $\tau:\aA \to \aG$ maximizing $\prod_{i\in \aA} (v_i(\tau(i)))^{w_i}$ and set $H \coloneqq \tau(A), J \coloneqq \aG \setminus H, \aAn:=\{i\in\aA:\, v_i(\aJ)>0\}$ \label{algo:phase1} \; 
\For{$i\in \aAn$}{$\ell(i) \gets \argmax \{ v_i(\ell): \ell \in J \}$ \; 
Define  $\bar{v}_i(S) \coloneqq v_i(\ell(i)) + v_i(S)$ \;}
$\mathcal{R}=(R_i)_{i\in \aA} \coloneqq  $\pLocalSearch{$J,(\bar v_i)_{i\in \bar{\aA}}$} \label{algo:phase2} \;
Find a matching $\sigma: \aA \to  H$ maximizing $\prod_{i \in \aA} (v_i(R_i + \sigma(i)))^{w_i}$ \label{algo:phase3}\;
\Return{$\cS=(R_i + \sigma(i))_{i\in \aA}$}\;
\caption{Approximating the submodular NSW problem}\label{alg:NSW-template}
\end{algorithm}

\paragraph{Phase 1: Initial matching}
We find an optimal assignment of one item to each agent, i.e.,~a matching 
$\tau:\aA \to \aG$ maximizing $\prod_{i\in \aA} (v_i(\tau(i)))^{w_i}$. 
This can be done using a maximum-weight matching algorithm with weights
$w_i \log v_i(j)$ in the bipartite graph between $\aA$ and $\aG$ with edge set $\{(i,j):\, v_i(j)>0\}$.
If no matching of cardinality $n=|\aA|$ exists, then we can conclude that there is no allocation with positive NSW value, and  return an arbitrary allocation. For the rest of the paper, we assume there is a matching covering $\aA$, and let $ H \coloneqq \tau(A)$ be the set of matched items. 

\paragraph{Phase 2: Local search} In the second phase, we let $\aJ \coloneqq \aG\setminus H$ denote the set of items not assigned in the first phase.
Recall that we let $\aAn:=\{i\in\aA:\, v_i(\aJ)>0\}$ denote the set of 
agents that have a positive value for some item in $\aJ$.
For every $i\in\aA$,  we select
\[\ell(i) \in \argmax_{j\in \aJ} v_{i}(j) \]
 as a \emph{favorite} item of agent $i$ in $\aJ$ (for agents in $\aA\setminus \aAn$, this will be an arbitrary item with $v_i(\ell(i))=0$.) By submodularity, $v_{i}(\ell(i))>0$.
For each $i\in\aAn$, we define the endowed valuation function
 $\bar v_i : 2^\aJ \to \RR_{> 0}$ as 
 \[
 \bar v_i(S) \coloneqq v_i(\ell(i)) + v_i(S) \quad \forall S\subseteq \aJ\, .
 \]
 Thus, $\bar v_i(\emptyset)=v_{i}(\ell(i))$, and $\bar v_i(j) \le 2 \bar  v_i(\emptyset)$ for any $j\in\aJ$. Further, we set the accuracy parameter
\begin{equation}\label{eq:bar-eps}
\bar\eps\coloneqq \frac{\eps}{2m}\, .%
\end{equation}
Note that by the assumption $\eps\le 1$, the following bound holds:
\begin{equation}\label{eq:bar-eps-bound}
(1+\bar\eps)^m< 1+\eps\, .
\end{equation}
Our local search starts with allocating all items to a single agent  in $\aAn$.
As long as moving one item to a different agent increases the potential function
\[
\prod_{i\in \aAn} 
\left(\bar v_i(R_i)
\right)^{w_i}
\]
by at least a factor $(1+\bar\eps)$, we perform such an exchange. Phase 2 terminates when no more such exchanges are possible, and returns the current allocation. For all agents $i\in\aA\setminus \aAn$, we let $R_i=\emptyset$.

\begin{algorithm}[htbp]
\DontPrintSemicolon
$R_{k} \gets \aJ$ for some $k\in \aAn$ and $R_i\gets \emptyset$ for $i\in \aA \setminus \{k$\}\;
\While{$\exists i,k\in\aAn$ and $j\in R_i$ such that $\left( \frac{\bar v_i(R_i - j)}{\bar v_i(R_i)} \right)^{w_i}
\cdot \left( \frac{\bar v_k(R_k + j)}{\bar v_k(R_k)} \right)^{w_k} > 1+\bar\eps$}{
$R_i \gets R_i - j$ and $R_k \gets R_k + j$\;}
\Return{$\cR:=(R_i)_{i\in \aA}$}
\caption{{\texttt{LocalSearch($J,(\bar v_i)_{i\in \aAn}$)}}}\label{alg:localSearch}
\end{algorithm}

\paragraph{Phase 3: Rematching}
In the final phase, we match the items in $H$ to the agents optimally, considering allocation $\cR = (R_i)_{i\in\aA}$ of $J$. This can be done by again solving a maximum-weight matching problem, now with weights $w_{ij} = w_i \log v_i(R_i+j)$. 
\medskip

In the remainder of this section, we prove Theorem~\ref{thm:main}.
In Section~\ref{sec:rematching}, presents the main `Rematching Lemma' (Lemma~\ref{lem:rematching-estimation}) that asserts that Algorithm~\ref{alg:NSW-template} has value at least corresponding to the following solution: for every agent, take the best of $R_i$, $\ell(i)$, and $\pi(i)$, where $\pi(i)$ is a matching assigning to every agent their most valuable item from $H$ in an optimal solution. Note that $\max\{v_i(R_i),v_i(\ell(i))\}\ge \frac{1}{2} \bar v_i(R_i)$, hence, this enables us to relate the value of the solution to the local search with respect to the endowed valuations.

 In Section~\ref{sec:local}, we formulate simple properties of approximate local optimal solution found in Phase 2. 
 We then compare the allocation of items in $J$ in an optimal solution to $(R_i)_{i\in \aA}$  via a marginal pricing argument. This is done separately for the symmetric case (Section~\ref{sec:symmetricLS}) and the asymmetric case (Section~\ref{sec:AsymmetricAnalysisSimple}). The arguments follow the same lines, but the marginal prices are defined differently in the two settings. These together complete
the proof of Theorem~\ref{thm:main}. 

\subsection{The rematching lemma}
\label{sec:rematching}
A crucial step in our algorithm relates the value of the assignment output by Algorithm~\ref{alg:NSW-template} to the contribution of $H$ to an optimal allocation. Suppose that the optimal solution is
$(H_i\cup S_i)_{i\in \aA}$, where $H_i\subseteq H$ and $S_i\subseteq J$. This optimal solution does not necessarily  assign $H$ via a matching---i.e., the cardinalities of
$H_i$ could be arbitrary. Given $(H_i)_{i\in \aA}$, we define a matching $\pi: \aA \to H$
as follows: If $ H_i \neq \emptyset$, let $\pi(i) \in \arg\max_{j\in H_i} \bar v_i(j)$. If $H_i = \emptyset$, then define $\pi(i)$
arbitrarily so that $\pi$ is a perfect matching between $\aA$ and $H$.

A key lemma is the following. 
We note that the argument corresponds to a more general statement about bipartite matchings which might be of independent interest. We present and discuss this statement in Appendix~\ref{sec:laci-matching}.

\begin{lemma}[Rematching Lemma]\label{lem:rematching-estimation}
The value obtained by Algorithm~\ref{alg:NSW-template} is at least 
\[
\prod_{i\in \aA} \left(\max\{v_i(\pi(i)), v_i(\ell(i)), v_i(R_i)\}\right)^{w_i}.
\]
\end{lemma}

\begin{proof}
The allocation returned by our algorithm is $(R_1 + \sigma(1), R_2 + \sigma(2), \ldots, R_n + \sigma(n))$ where $(R_1,\ldots,R_n)$ is the partition of $X \setminus H$ found by the local search, and $\sigma:A \to H$ is the matching found in the last stage of our algorithm. Instead, let us consider the following allocation, which is more restricted: Given $H$ and $(R_1,\ldots,R_n)$, each agent can choose either their set $R_i$, or one of the items in $H$, but not both. We prove that the best such allocation already satisfies the guarantee of the lemma.

Without loss of generality, let us denote the set of agents as $\aA = \{1,2,\ldots,n\}$, with $\aAn=\{1,2,\ldots,\bar n\}$ for $\bar n\le n$. 
Let us construct the following bipartite graph:  one side includes $A$, and the other side has $n+2\bar n$ nodes, labeled $\{h_1,\ldots,h_n, \ell_1,\ldots, \ell_{\bar n}, r_1, \ldots, r_{\bar n} \}$. The nodes $h_1,\ldots, h_n$ represent the elements of $H$. The nodes $\ell_1,\ldots,\ell_{\bar n}$ represent the items $\ell(1), \ldots, \ell(\bar n)$ --- however, $\ell_1, \ldots, \ell_{\bar n}$ are distinct nodes even if some items $\ell(i)$ coincide. Finally, the nodes $r_1,\ldots, r_{\bar n}$ represent the sets $R_1, \ldots, R_{\bar n}$. We define an edge $(i, h_j)$ for each $i \in \aA$ and $h_j \in H$ with $v_i(h_j)>0$, of weight $w_i\log v_i(h_j)$. We also define an edge $(i, \ell_i)$ of weight $w_i \log v_i(\ell(i))$, and an edge $(i, r_i)$ of weight $w_i \log v_i(R_i)$ for every $i\in \aAn$.

\begin{figure}[htb!]
\begin{flushleft}
\caption{Proof of Lemma~\ref{lem:rematching-estimation}, illustration: the alternating path $h_3 - 3 - h_4 - 4 - \ell_4$ can be used to improve either $\rho$ or $\tau$.}
\begin{tikzpicture}[scale=1.0]
	\node at (-7,6) {};
    \foreach \i in {1, 2, 3, 4} {
        \node[rectangle, draw, fill=blue!40] (A\i) at (0,5-\i) {\large $\i$};
        \node[circle, draw, fill=red!40] (H\i) at (2.5, 6-0.7*\i) {\small $h_{\i}$};
        \node[circle, draw, fill=green!50] (L\i) at (2.5, 3-0.7*\i) {\small $\ell_{\i}$};
	 }

    \foreach \i in {1, 2, 3, 4} {
	\node [circle, draw, fill=red!20]  (S\i) at (-2,5-\i)  {$R_{\i}$};
	 }
    
    \node at (0, 4.8) {$A$};
    \node at (3.2, 4.2) {$H$};

    \foreach \i in {1, 2, 3, 4} {
	    \draw [gray] (A\i) -- (S\i);
	    \draw [gray] (A\i) -- (L\i);
	    
	    \foreach \j in { 1, 2, 3, 4} {
	    	\draw [gray] (A\i) -- (H\j);
		}
	}
    	
	\draw [thick, red] (A1) -- (H2);
	\draw [thick, red] (A2) -- (S2);
	\draw [thick, red] (A3) -- (H4);
	\draw [thick, red] (A4) -- (L4);	
 	\node [red] at (1.25, 0.2) {$\rho$};
	
	\draw [thick, blue] (A1) -- (H1);
	\draw [thick, blue] (A2) -- (H2);
	\draw [thick, blue] (A3) -- (H3);	
	\draw [thick, blue] (A4) -- (H4);
 	\node [blue] at (1.25, 5.0) {$\tau$};	
\end{tikzpicture}
\end{flushleft}
\end{figure}

Let $\rho$ be a maximum-weight matching of cardinality $|A|$ in this bipartite graph. Clearly, its weight is at least $\sum_{i=1}^{n} w_i\log \max \{ v_i(\pi(i)), v_i(\ell(i)), v_i(R_i) \}$, because one possible matching is to choose for each $i \in A$ the edge corresponding to the maximum of the 3 quantities. An issue with this matching is that it does not correspond to a feasible allocation, due to possible overlaps between the items $\ell(i)$ and the sets $R_i$. 

The key claim is that there is a maximum-weight matching in this graph which {\em does not use any of the edges} $(i,\ell_i)$. We prove this as follows: Suppose that $\rho$ is a maximum-weight matching that uses as few of the $(i,\ell_i)$ edges as possible. Recall also the matching $\tau$ constructed in the initial stage, which is the maximum-weight matching between $A$ and the full item set $\aG$. In our graph, $\tau$ matches $A$ to $H = \{ h_1, \ldots, h_n\}$. Suppose that $\rho$ uses an edge $(i^*, \ell_{i^*})$ and consider the alternating path $P$ in $\rho \Delta \tau$, starting from this edge ($\tau$ does not match the node $\ell_{i^*}$, so this is indeed the start of an alternating path). Except for the node $\ell_{i^*}$, the alternating path hops between $A$ and $H$, because starting from $\ell^*$, an edge of $\rho$ takes us to $A$, then an edge of $\tau$ takes us to $H$, etc. Eventually, the alternating path terminates at a node of $H$ which is not matched by $\rho$ (whereas each node of $A$ is matched by both $\tau$ and $\rho$). Now, consider two cases. \\
(1) The weight of the $\tau$-edges on this alternating path $P$ is at least as large as the weight of the $\rho$-edges on $P$. Then we can perform a swap, $\rho' = \rho \Delta P$. This defines a matching $\rho'$  in our graph whose weight is at least as large as that of $\rho$, and the number of $(i, \ell_i)$ edges has decreased---a contradiction. \\
(2) The weight of the $\tau$-edges on $P$ is strictly less than the weight of the $\rho$-edges on $P$. Then we can perform the swap $\tau' = \tau \Delta P$. Note that $\tau'$ is a valid matching of items to agents, because it uses only $H$ and one item $\ell(i^*)$, which is outside of $H$, so no conflict can arise here. Thus we obtain a matching $\tau':A \to X$  of weight strictly larger than $\tau$, which means $\tau$ was not a maximum-weight matching in the initial stage---a contradiction.

We conclude that there is a maximum-weight matching in our graph which does not use any $(i, \ell_i)$ edge, and hence it corresponds to a feasible allocation of either $R_i$ or an element of $H$ to each agent.
\end{proof}

\subsection{Characterizing local optima}\label{sec:local}
\label{sec:localSearch}
In order to upper bound the optimal NSW value to the value obtained in Algorithm~\ref{alg:NSW-template} via Lemma~\ref{lem:rematching-estimation}, we also need to relate $S_i$, the part of the optimal solution in $\aJ$, to $R_i$ and $\ell(i)$. 
Note that Lemma~\ref{lem:rematching-estimation} does not use the local optimality of $R_i$ but works for an arbitrary partition of $J$. In this subsection, we focus on properties of locally optimal solutions.

Recall that we work with the item set
$\aJ$, set of agents $\aAn$, favourite items $\ell(i)$, and endowed valuations  $\bar{v}_i(S) = v_i(\ell(i)) + v_i(S)$.

\begin{definition}[$\bar\eps$-local optimum]
A complete allocation $\cR = (R_i)_{i\in \aA}$ is an \emph{$\bar\eps$-local optimum with respect to valuations $\bar v_i$}, if for all pairs of different agents $i, k \in \aAn$ and all $j\in R_i$ it holds that
\[
\left( \frac{\bar v_i(R_i - j)}{\bar v_i(R_i)} \right)^{w_i}
\cdot \left( \frac{\bar v_k(R_k + j)}{\bar v_k(R_k)} \right)^{w_k} \le  1+\bar\eps\,.
\]
\end{definition}
A 0-local optimum will be simply called a \emph{local optimum}. 

\begin{lemma}
Consider an NSW instance with submodular valuations, and let $\eps>0$.  The subroutine \pLocalSearch{$J,(\bar v_i)_{i\in \bar{\aA}}$} returns an $\bar{\eps}$-local maximum with respect to the endowed valuations $\bar v_i$ in $O\left(\frac{m}{\eps} \log m\right)$ exchange steps.
\end{lemma} 
\begin{proof}
It is immediate that if the algorithm terminates, it returns  
an $\bar{\eps}$-local maximum.
Recalling that $\bar v_i(j)\le 2\bar v_i(\emptyset)$ for any $j\in\aJ$, submodularity implies 
\[
\bar v_i(\aJ) < \sum_{i\in \aJ} v_i(j)+v_i(\ell(i))\le (|\aJ|+1) = v_i(\ell(i))= (|\aJ|+1) \bar v_i(\emptyset)\le m \bar v_i(\emptyset)
\]
for every $i\in\aAn$, using also that $|J|\le m-n\le m-1$.
Hence, 
\[\prod_{i\in \aAn}
\left({\bar v_i(\aJ})\right)^{w_i} \le m\prod_{i\in \aAn}\left({\bar v_i(\emptyset)}\right)^{w_i}\, ,\]
and therefore the product $\prod_{i\in \aAn}{\bar v_i(R_i)^{w_i}}$ may grow by at most a factor of $m$ throughout all exchange steps. Every swap increases this product by at least a factor  $(1+\bar\eps)$. 
Thus, the total number of swaps is bounded by $\log_{1+\bar\eps} m = m \log_{1+\eps} m = O\left(\frac{m}{\eps} \log m\right)$.
\end{proof}

We need the following two properties of submodular functions.
\begin{proposition}
\label{prop:submodularFraction}
Let $\bar{v}: 2^\aJ\to \RR_{> 0}$ be a monotone submodular function. 
Let $S\subseteq T\subseteq \aJ$ and $j \in \aJ$.
Then,
\[
\frac{\bar v(T+j)}{\bar v(T)} \le \frac{\bar v(S+j)}{\bar v(S)}\,.
\]
\end{proposition}

\begin{proof}
By the monotonicity and submodularity of $\bar v$, we have
\[
\begin{aligned}
\frac{\bar v(T+j)}{\bar v(T)} 
= \frac{\bar v(T) + \bar v(T+j) - \bar v(T)}{\bar v(T)}
& \le \frac{\bar v(S) + \bar v(T+j) - \bar v(T)}{\bar v(S)}  \\
&\le \frac{\bar v(S) + \bar v(S+j)- \bar v(S)}{\bar v(S)} 
= \frac{\bar v(S+j)}{\bar v(S)}.
\end{aligned}
\qedhere
\]
\end{proof}
The next bound will be used in the marginal pricing arguments.
\begin{proposition}
\label{prop:marginals}
Let $\bar{v}_i: 2^\aJ\to \RR_{> 0}$ be the submodular endowed valuation for an agent $i\in \bar{\aA}$.
For any $j \in R$, 
$$ \bar{v}_i(R-j) \ge  \sum_{k \in R} (\bar{v}_i(R) - \bar{v}_i(R-k)).$$
\end{proposition}

\begin{proof}
Let us denote $R-j \coloneqq \{r_1, \dots, r_s\}$. By submodularity, we have
\[
\begin{aligned}
\bar v_i(R-j) &= \bar v_i(\emptyset) + \sum_{k=1}^s (\bar v_i(\{r_1,\dots, r_k\})-\bar v_i(\{r_1,\dots, r_{k-1}\}))\\
&\ge \bar v_i(\emptyset) + \sum_{k=1}^s (\bar v_i(R) - \bar v_i(R-r_{k})) \geq  \sum_{k \in R} (\bar v_i(R) - \bar v_i(R-r_{k}))
\end{aligned}
\]
where in the last step, we used the fact that  $\bar v_i(\emptyset) = v_i(\ell(i))\ge v_i(j) \ge \bar v_i(R)-\bar v_i(R-j)$\,. %
\end{proof}

We analyze our local search in slightly different ways in the symmetric case and the general asymmetric case. We consider the symmetric case first.

\subsection{Local optimality analysis for symmetric NSW}
\label{sec:symmetricLS}
Recall that $\bar\eps=\eps/(2m)$. We define another accuracy parameter for the analysis as 
\begin{equation}\label{eq:hat-eps}
\hat \eps \coloneqq (1+\bar \eps)^n - 1\, .
\end{equation}
Thus, $1+ \hat \eps= (1+\bar \eps)^{n} \le (1+\bar \eps)^m = 1+\eps$ from \eqref{eq:bar-eps-bound} and  $n \le m$. 
In what follows, we complete the proof of Theorem~\ref{thm:main} for the symmetric case, i.e., $w_1 = \ldots = w_n = 1/n$, showing the approximation factor $4+\varepsilon$. We use a marginal pricing argument, associating a price $p_j$ with each item that reflects the marginal utility of the agent who owns it. The two main arguments show that {\em (a)} the total price of the items is at most $n$; and {\em (b)} the relative utility increase if an agent receives a bundle $T$ on top of their current bundle can be bounded by the price $p(T)$ (and a small error term).

Let $j\in \aJ$ and let $i\in \aAn$ be the agent such that $j\in R_i$. 
We define the \emph{price} of $j$ as
\[
p_j \coloneqq \frac{\bar v_{i}(R_{i})}{\bar v_{i}(R_{i} - j)} - 1 = \frac{\bar v_{i}(R_{i})- \bar v_{i}(R_{i} - j)}{\bar v_{i}(R_{i} - j)}  \,.
\]
The first lemma shows that the \emph{spending} of each agent $i$, $p(R_i) = \sum_{j \in R_i} p_j$, is at most $1$, and thus the total value of the items is at most $n$.

\begin{lemma}[Bounded spending]
\label{lem:sym-spending-bound}
Let $\cR=(R_i)_{i\in \aA}$ be an $\bar\eps$-local optimum with respect to the endowed valuations $\bar v_i$. Then, $p(R_i) \le 1$ for every agent $i\in \aAn$, and consequently, $p(J) \leq |\aAn|\le n$.
\end{lemma}

\begin{proof}
From the definition of the prices $p_j$, and by Proposition~\ref{prop:marginals}, we have
\[
p(R_i) = \sum_{j\in R_i} \frac{\bar v_i(R_i) - \bar v_i(R_i-j)}{\bar v_i(R_i-j)} \le \frac{\sum_{j\in R_i} (\bar v_i(R_i) - \bar v_i(R_i-j))}{\min_{k\in R_i} \bar v_i(R_i-k)} \leq 1\,.
\]
Since $(R_1,\ldots,R_n)$ is a partition of $J$ (every item is allocated throughout the local search), we have
\[ p(J) = \sum_{j \in J} p_j = \sum_{i \in \aAn} \sum_{j \in R_i} p_j \leq |\aAn|\,.\qedhere \]
\end{proof}

\begin{lemma}
\label{lem:sym-local-opt}
Given an {$\bar\eps$-local} optimum $\cR=(R_i)_{i\in \aAn}$, and the prices $p_j$ defined as above, for every $k\in\aAn$ and  $T\subseteq \aJ$, it holds that
\[
\frac{\bar v_k(R_k \cup T)}{\bar v_k(R_k)} \le 1 + p(T) + 2 {\hat\eps}|T| \,.\]
\end{lemma}
\begin{proof}
First, note that
\begin{equation}\label{eq:sym-price-bound}
p_j\le 1\quad \forall j\in \aJ\, .
\end{equation}
This follows already by Lemma~\ref{lem:sym-spending-bound}; one can also see it directly since by construction of $\bar v_i$ and submodularity, $\bar v_i(R_i) - \bar v_i(R_i-j) \leq v_i(j) \leq \bar v_i(\emptyset) \leq \bar v_i(R_i-j)$. Hence, $p_j = \frac{\bar v_{i}(R_{i})- \bar v_{i}(R_{i} - j)}{\bar v_{i}(R_{i} - j)} \leq 1$.

Let $i\in \aAn$ such that $j\in R_i$.
From the $\bar\eps$-optimality of $\cR$, using also the definition of $\hat\eps$  and \eqref{eq:sym-price-bound}, we get 
\begin{equation}\label{eq:sym-marginal}\frac{\bar v_k(R_k+j)}{\bar v_k(R_k)} \leq (1+\bar{\eps})^n \frac{\bar v_{i}(R_{i})}{\bar v_{i}(R_{i} - j)} = (1+\hat\eps) (1+p_j)\le 1+p_j+2\hat\eps\, ,
\end{equation} because otherwise we could swap item $j$ to agent $k$; recall that $w_k=w_i=1/n$. Using  submodularity together with this bound, we have
\[
\begin{aligned}
\frac{\bar v_k(R_k \cup T)}{\bar v_k(R_k)} & \le
 \frac{\bar v_k(R_k) + \sum_{j \in T} (\bar v_k(R_k + j) - \bar v_k(R_k))}{\bar v_k(R_k)}\\
&\le 1 + \sum_{j \in T} (p_j+2\hat\eps) \leq p(T)+1 + 2\hat\eps|T|\, .
\end{aligned}
\]
\end{proof}
As a small detour to illustrate the usefulness of the above bound, we show that local search  gives a $2$-approximation of the optimal Nash social welfare with respect to the {\em endowed valuations} for the agents in $\aAn$ that have positive valuation for $J$. 

\begin{proposition}\label{prop:firstWelfare}
Let $\cR= (R_i)_{i\in \aA}$ be a local optimum returned by Algorithm~\ref{alg:localSearch} and $\cS=(S_i)_{i\in \aA}$ be an optimal NSW allocation  of $\aJ$ with respect to the endowed submodular valuations $\bar v_i$.
Then 
\[
\left(\prod_{i\in \aAn} \bar v_i(R_i)\right)^{1/n}  \ge \frac{1}{2} \cdot \left(\prod_{i\in \aAn} \bar v_i(S_i)\right)^{1/n} \,. 
\]
\end{proposition}
\begin{proof}
By Lemma~\ref{lem:sym-spending-bound}, $\sum_{i\in\aAn}p(S_i) = p(J)\le n$.
Recall that local optimality corresponds to $\bar\eps=\hat\eps=0$.
By Lemma~\ref{lem:sym-local-opt} and the AM-GM inequality,
\begin{align*}
    \prod_{i\in\aAn} \bar v_i(S_i)
& \le \prod_{i\in\aAn} \bar v_i(R_i\cup S_i)
\le \prod_{i\in\aAn} (1 + p(S_i)) \bar v_i(R_i) \\
& \leq \left(\frac1n\sum_{i\in\aAn}(1 + p(S_i))\right)^n \prod_{i\in\aAn} \bar  v_i(R_i) \\
& \le 2^n \cdot  \prod_{i\in\aAn} \bar  v_i(R_i) \,. \qedhere
\end{align*}

\end{proof}
Proposition~\ref{prop:firstWelfare} is included solely for the intuition and we do not use it directly: we still need to take the items in $H$ into consideration and move from the endowed valuations $\bar v_i$ to the original valuations $v_i$.
In what follows, we relate the locally optimal allocation of $J$, $(R_i)_{i\in {\aA}}$ to the optimal solution $(H_i\cup S_i)_{i\in \aA}$ with $H_i\subseteq H$ and $S_i\subseteq J$ corresponding to the parts of the optimal bundle in $H$ and $J$, respectively.

\begin{lemma}\label{lem:sym-price-bound}
Let $\cR=(R_i)_{i\in \aA}$ be an $\bar\eps$-local optimum with respect to the endowed valuations $\bar v_i$, and let $i\in \aAn$. For any set $H_i\subseteq H$, $S_i\subseteq J$, and $\pi(i)\in\argmax_{j\in H_i} \bar v_i(j)$ if $H_i\neq \emptyset$ and $\pi(i)\in H$ arbitrary otherwise. 
\[
\bar v_i(R_i\cup S_i \cup H_i) \le (1 + p(S_i) + 2\hat\eps|S_i|) \, \bar v_i(R_i) + |H_i| \, v_i(\pi(i))\, . 
\]
\end{lemma}

\begin{proof}
By submodularity and the choice of $\pi(i)$, $v_i(H_i) \leq \sum_{j \in H_i} v_i(j) \leq |H_i| v_i(\pi(i))$.
 By  Lemma~\ref{lem:sym-local-opt} and submodularity again, the lemma follows.
\end{proof}

We are ready to prove the approximation guarantee in the symmetric case.

\begin{lemma}\label{lem:rematching}
Let $\eps\ge0$, $\bar\eps=\eps/(2m)$, and $\cR=(R_i)_{i\in \aA}$ be an $\bar\eps$-local optimum with respect to the endowed valuations $\bar v_i$.
Then, there exists a matching $\rho: \aA \to H$ such that
\[
\NSW(\cR, \rho) \ge \frac{\OPT}{4(1+\eps)}  \,.
\]
\end{lemma}

\begin{proof}
Consider the optimal allocation $(H_i\cup S_i)_{i\in \aA}$ and $\pi(i)$, $i\in\aA$ as above. Denote 
\[
V_i \coloneqq \max \{ v_i(\pi(i)), v_i(\ell(i)), v_i(R_i) \}\, ,\]
 and recall that according to Lemma~\ref{lem:rematching-estimation},  our algorithm achieves value at least $\left( \prod_{i=1}^{n} V_i \right)^{1/n}$.
 First, consider an agent $i\in\aAn$.
By Lemma~\ref{lem:sym-price-bound},
$$ \bar{v}_i(S_i \cup H_i) \leq \bar{v}_i(R_i \cup S_i \cup H_i) \leq (1 + p(S_i) + 2\hat\eps|S_i|) \, \bar v_i(R_i) + |H_i| v_i(\pi(i)).$$
Equivalently, since $\bar{v}_i(S) = v_i(S) + v_i(\ell(i))$,
\begin{equation*}
\begin{aligned}
v_i(S_i \cup H_i) &\leq v_i(R_i) + (p(S_i) + 2\hat\eps|S_i|)(v_i(R_i) + v_i(\ell(i))) + |H_i| v_i(\pi(i))\\
&\leq (1 + 2 p(S_i) + 4\hat\eps|S_i| + |H_i| ) V_i \, .
\end{aligned}
\end{equation*}
The same bound also holds for $i\in \aA\setminus\aAn$, noting that $v_i(S_i\cup H_i)=v_i(H_i)\le |H_i| v_i(\pi(i))$.
Hence, by Lemma~\ref{lem:rematching-estimation}, we get
$$ \OPT = \left( \prod_{i=1}^{n} v_i(S_i \cup H_i) \right)^{1/n} \leq \prod_{i=1}^{n} \left( (1 + 2 p(S_i) + 4\hat\eps|S_i| + |H_i| ) V_i \right)^{1/n}.$$
By the AM-GM inequality, 
\begin{eqnarray*}
\OPT  &  \leq  & \prod_{i=1}^{n} \left( 1 + 2 p(S_i) + 4\hat\eps|S_i| + |H_i| \right)^{1/n} \left(\prod_{i=1}^{n} V_i \right)^{1/n} \\
 & \leq & \left(\frac{1}{n} \sum_{i=1}^{n} \left( 1 + 2 p(S_i) + 4\hat\eps|S_i| + |H_i| \right) \right) \left(\prod_{i=1}^{n} V_i \right)^{1/n}  \\
 & \leq & 4\left(1+\frac{m\hat\eps}{n}\right) \left(\prod_{i=1}^{n} V_i \right)^{1/n}
 \end{eqnarray*}
using the facts that $\sum_{i=1}^{n} p(S_i) \leq n$, $\sum_{i=1}^{n} |H_i| \leq n$ and $|S_i|\le m$. The proof is complete since $\left(1+\frac{m\hat\eps}{n}\right)\le (1+\hat\eps)^{m/n}=(1+\bar \eps)^m<(1+\eps)$ by the definition of $\hat\eps$ and by \eqref{eq:bar-eps-bound}.
\end{proof}

\subsection{Analysis for asymmetric NSW}
\label{sec:AsymmetricAnalysisSimple}
We now modify the above analysis to the case of general 
weights $w_i$. The main difference is in the defintion of the prices $p_j$.
For $\bar\eps=\eps/(2m)$, let 
 $\cR=(R_i)_{i\in \aA}$ be an $\bar\eps$-local optimum with respect to the endowed valuations $\bar v_i$.
Let $j\in \aJ$ and let $i\in \aAn$ be the agent such that $j\in R_i$. 
We define the \emph{price} of $j$ as
\[
p_j \coloneqq w_i \log \frac{\bar v_{i}(R_{i})}{\bar v_{i}(R_{i} - j)}  \,.
\]
Note that in the symmetric case we replaced the expression of the form $w_i \log \alpha$ by $\alpha-1$. This allowed a tighter bound; however, the proof of Lemma~\ref{lem:sym-local-opt} would not extend to the asymmetric case with such a price definition. In Lemma~\ref{lem:localOptima}, the bound includes an exponential expression of the price instead of a linear one.
We first show the analogue of Lemma~\ref{lem:sym-spending-bound}. Note that the total spending is now bounded as $p(J)\le 1$ rather than $p(J)\le n$; this is mainly due to the incorporation of the weights $w_i$ in the definition of the prices. 

\begin{lemma}[Bounded spending]
\label{lem:spendingBound}
For an $\bar\eps$-local optimum $\cR = (R_i)_{i \in \aAn}$ and prices $p_j$ defined as above, 
$p(R_i) \le w_i$ for every agent $i\in \aAn$, and consequently, $p(J)\le 1$.
\end{lemma}
\begin{proof}
From the definition of $p_j$, we have
\[
\begin{aligned}
p(R_i)&=w_i \sum_{j \in R_i} \log \frac{\bar v_i(R_i)}{\bar v_i(R_i-j)} \le 
w_i \sum_{j \in R_i} \frac{\bar v_i(R_i) - \bar{v}_i(R_i-j)}{\bar v_i(R_i-j)} \\
& \le w_i \frac{\sum_{j \in R_i} \bar v_i(R_i) - \bar{v}_i(R_i-j)}{\min_{k\in R_i}\bar v_i(R_i-k)}
\le w_i\, .
\end{aligned}
\]
In the first inequality, we used $\log x \leq x-1$, and in the third inequality we used Proposition~\ref{prop:marginals}.
Adding up the prices over all the sets $R_i$, whose union is $J$, we obtain $p(J) = \sum_{i \in \aAn} p(R_i) \leq \sum_{i \in \aAn} w_i \leq 1$.
\end{proof}

\begin{lemma}
\label{lem:localOptima}
For endowed valuations $\bar v_k$, an $\bar\eps$-local optimum $\cR = (R_i)_{i \in \aAn}$ and prices $p_j$ defined as above, for every item $j \in R_i$ and every agent $k \in \aAn$, we have 
\[
\frac{\bar v_k(R_k + j)}{\bar v_k(R_k)} \le (1+\bar\eps)^{1/w_k} \ee^{p_j / w_k} \, .
\]
Moreover, if the endowed  valuation $\bar v_k$ is submodular, for all $T\subseteq \aJ$, we have
\[
\frac{\bar v_k(R_k \cup T)}{\bar v_k(R_k)} \le (1+\bar\eps)^{|T|/w_k} \cdot \ee^{p(T)/ w_k} \,.\]
\end{lemma}

\begin{proof}
By definition, $\ee^{p_j/w_i} = \frac{\bar v_{i}(R_{i})}{\bar v_{i}(R_{i} - j)}$. 
If $k = i$ the first statement is trivial. 
Otherwise, for $k\neq i$, the first statement follows from  the $\bar\eps$-optimality of $\cR$; if false, we would swap item $j$ to agent $k$.  

For the second statement, let us denote $T = \{t_1,t_2,\ldots,t_{|T|}\} \subseteq \aJ$. Since $\bar v_k$ is submodular, by Proposition~\ref{prop:submodularFraction} we have
\[
\begin{aligned}
\frac{\bar v_k(R_k \cup T)}{\bar v_k(R_k)} = 
\prod_{a=1}^{|T|} \frac{\bar v_k(R_k \cup \{t_1, \dots, t_{a}\})}{\bar v_k(R_k \cup \{t_1, \dots, t_{a-1}\})}
&\le \prod_{a=1}^{|T|} \frac{\bar v_k(R_k + t_{a} )}{\bar v_k(R_k)}\\
&\le \left(1+\bar\eps\right)^{|T|/w_k} \ee^{p(T)/w_k}  \,. 
\end{aligned}
\qedhere\]
\end{proof}

Similarly to Lemma~\ref{lem:sym-price-bound}, we have the following lemma for the bound of optimum solutions using prices and values of singletons in $H$.

\begin{lemma}\label{lem:asym-price-bound}
Let $\cR=(R_i)_{i\in \aA}$ be an $\bar\eps$-local optimum with respect to the endowed valuations $\bar v_i$. For $i\in\aAn$, let  $H_i\subseteq H$, $S_i\subseteq J$,  $\pi(i)\in \arg\max_{j\in H_i} v(j)$ if $H_i\neq\emptyset$ and $\pi(i)\in H$ arbitrary otherwise. Then,
\[
\bar v_i(R_i\cup S_i \cup H_i) \le (1+\bar\eps)^{|S_i|/w_i} \ee^{p(S_i) /w_i}  \bar v_i(R_i) + |H_i|v_i(\pi(i)). 
\]
\end{lemma}

\begin{proof}
 We have
\[
\begin{aligned}
	\bar v_i(R_i\cup S_i \cup H_i) &= v_i(R_i\cup S_i\cup H_i)+v_i(\ell(i))\\ 
	&\le v_i(R_i\cup S_i)+v_i(H_i)+v_i(\ell(i))\\
	&\le \bar v_i(R_i\cup S_i)+  \sum_{j \in H_i} v_i(j)\\
	&\le (1+\bar\eps)^{|S_i|/w_i} \ee^{p(S_i) /w_i} +  |H_i| v_i(\pi(i))\, .
\end{aligned}
\]
Here, the first and the second inequality uses submodularity. The third inequality uses Lemma~\ref{lem:localOptima}
and the choice of $\pi(i)$.
\end{proof}

Our final approximation bound follows by combining this with Lemma~\ref{lem:rematching-estimation}.
\begin{lemma}\label{lem:gen-rematching}
Let $\eps\ge0$, $\bar\eps=\eps/(2m)$, and let $\cR=(R_i)_{i\in \aA}$ be an $\bar\eps$-local optimum with respect to the endowed valuations $\bar v_i$.
Then, there exists a matching $\rho: \aA \to H$ such that
\[
\NSW(\cR, \rho) \ge \frac{\OPT}{(2+nw_{\max})\ee(1+\eps)}  \,.
\]
\end{lemma}

\begin{proof}
Let the optimal allocation to agent $i\in\aA$ be $S_i \cup H_i$, as above, with $\pi(i)\in \argmax_{j\in H_i} v(j)$ and $\pi(i)\in H$ arbitrary otherwise. Let
\[
V_i \coloneqq \max \{ v_i(\pi(i)), v_i(\ell(i)), v_i(R_i) \}\, .
\]
Thus, Lemma~\ref{lem:rematching-estimation} shows that Algorithm~\ref{alg:NSW-template} returns a solution with value at least $\prod_{i\in\aA} V_i^{w_i}$.
 By Lemma~\ref{lem:asym-price-bound}, for $i\in\aAn$,
$$ \bar{v}_i(S_i \cup H_i) \leq \bar{v}_i(R_i \cup S_i \cup H_i) \leq (1+\bar\eps)^{|S_i|/w_i} \ee^{p(S_i) /w_i}  \bar v_i(R_i) + |H_i|v_i(\pi(i))
, .$$
Equivalently, since $\bar{v}_i(S) = v_i(S) + v_i(\ell(i))$,
$$ 
\begin{aligned}v_i(S_i \cup H_i) 
&\leq v_i(R_i) + ((1+\bar\eps)^{|S_i|/w_i} \ee^{p(S_i) /w_i} - 1)(v_i(R_i) + v_i(\ell(i))) + |H_i| v_i(\pi(i))\\
&\leq  \left( 2  (1+\bar\eps)^{|S_i|/w_i} \ee^{p(S_i) /w_i} - 1 + |H_i|\right)V_i\, .
\end{aligned}$$
The same bound trivially holds for agents $i\in \aA\setminus\aAn$.
Hence, 
$$ \OPT =  \prod_{i=1}^{n} \left(v_i(S_i \cup H_i) \right)^{w_i} \leq \prod_{i=1}^{n} \left(  2  (1+\bar\eps)^{|S_i|/w_i} \ee^{p(S_i) /w_i}  + |H_i| ) V_i \right)^{w_i}.$$
We now get
\begin{eqnarray*}
\OPT  
&  \leq  & \prod_{i=1}^{n} \left(2 (1+\bar\eps)^{|S_i|/w_i} \ee^{p(S_i) /w_i} + |H_i| \right)^{w_i} V_i^{w_i} \\
&  \leq  & \prod_{i=1}^{n} \left(2 + |H_i|\right)^{w_i} (1+\bar\eps)^{|S_i|} \ee^{p(S_i)} V_i^{w_i} \\
 & \leq & \left(\sum_{i=1}^n w_i(2 + |H_i|)\right) (1+\bar\eps)^m \ee \,\prod_{i=1}^{n} V_i^{w_i} \label{eq:lousy-estimate} \\ 
 & \leq & (2+nw_{\max})(1+\eps)\ee \,\prod_{i=1}^{n} V_i^{w_i}\, .
 \end{eqnarray*}
 Here, the second inequality follows since $(1+\bar\eps)^{|S_i|/w_i} \ee^{p(S_i) /w_i}\ge 1$, the third inequality follows by the AM-GM inequality 
and by  $\sum_{i=1}^{n} p(S_i) \leq 1$ (Lemma~\ref{lem:spendingBound}); the final inequality uses $\sum_{i=1}^{n} |H_i| \leq n$, and \eqref{eq:bar-eps-bound}.
\end{proof}

\begin{remark}[Natura~\cite{bento}]\label{rem:bento}\em
The bound in Theorem~\ref{thm:main} can be in fact strengthened to $\exp\big(\log(3n)+\sum_{i=1}^n w_i\log w_i\big)\cdot (\ee+\eps)$, a slight improvement over the bound in \cite{brown2024approximation} for the additive case.
This follows by strengthening the bound $(1+\varepsilon)(2+nw_{\max})$ in Lemma~\ref{lem:gen-rematching} to $3n(1+\varepsilon)\prod_{i\in\aA}w_i^{w_i}$. To see this, we only need to change the final estimate on $\prod_{i=1}^{n} \left(2 + |H_i|\right)^{w_i}$ in the proof. Here, we are given $h_i\ge 0$ such that $\sum_{i\in A} h_i\le n$, and need to upper bound $\prod_{i\in A}(2+ h_i)^{w_i}$. Instead of simply using AM-GM, we use AM-GM for a reweighted product:
\[
\begin{aligned}
\prod_{i\in\aA}(2 + h_i)^{w_i}=\prod_{i\in\aA}w_i^{w_i}\prod_{i\in\aA}\left(\frac{2+ h_i}{w_i}\right)^{w_i}
\le \prod_{i\in\aA}w_i^{w_i}\sum_{i\in \aA} (2 + h_i)\le 3n\prod_{i\in\aA}w_i^{w_i}\, . 
\end{aligned}
\]
\end{remark}

\subsection{Our techniques and comparison with previous approaches}\label{sec:overview-GHVLiV}
The ideas in this paper were gradually developed by a series of work by subsets of the authors, namely, the conference papers \cite{GHV20} and in \cite{LiV22}, and the joint paper \cite{GHLVV23}.  The present paper represents further simplifications of \cite{GHLVV23}, and can be thus seen as the joint journal version of these three papers.
At a high level, all three algorithms proceed in three phases, with Phases 1 and 3 (``matching and rematching'') being the same as outlined above. However, they largely differ in how the allocation $\cR$ of $J=\aG\setminus H$ is obtained in Phase 2.

The matching/rematching framework qas first used by
Garg, Husi\'c, and V\'egh \cite{GHV20}. For Phase 2 (items not included in the initial matching), they used a rational convex relaxation, based on the concave extension of Rado valuations.  After solving the relaxation exactly, they used combinatorial arguments to sparsify the support of the solution and construct an integral allocation. This relaxation and rounding exploits the combinatorial structure of Rado valuations, in particular, that the relaxation with the concave extension can be solved exactly. 

Li and Vondr\'ak \cite{LiV22} obtained the first polynomial-time algorithm for arbitrary submodular valuations. A significant challenge is that  for submodular functions,  the concave extension is NP-hard to evaluate. Instead, they worked with the multilinear extension, which can be evaluated with random sampling, but it is not concave. To solve the relaxation (approximately), they used an iterated continuous greedy algorithm. The allocation $\cR$ is obtained by independent randomized rounding of this fractional solution. Whereas the algorithm is simple, the analysis is somewhat involved. The main tool to analyze the rounding is the Efron--Stein concentration inequality; but this only works well if every item in the support of the fractional solution has bounded value. This is not true in general, and the argument instead analyzes a two-stage randomized rounding that gives a lower bound on the performance of the actual algorithm. First, a set of ``large'' fractional items is preserved, and a careful combinatorial argument is needed to complete the allocation.

Our approach for Phase 2 is radically different and much simpler. We do not use any continuous relaxation; instead, $\cR$ is obtained by discrete local search. Our discrete local search works with a certain modification of the valuation functions, where each agent is endowed with their largest singleton $\ell(i)$ ``for free''. 
Using these modified valuations, we can guarantee a high NSW value of an infeasible allocation $(R_i+\ell(i))_{i \in A}$ of $J$ in the analysis. 
Our pricing analysis of the local search is inspired by the \emph{conditional equilibrium} notion introduced by Fu, Kleinberg, and Lavi~\cite{DBLP:conf/sigecom/FuKL12ConditionalEqui}.
Conditional equilibrium is a relaxation of Walrasian equilibrium of indivisible goods:
this is an allocation and pricing where players cannot improve on their utility by adding new goods to their current bundle. They show that a conditional equilibrium  attains at least half of 
the maximal utilitarian welfare; this corresponds to an approximate version of the first welfare theorem for quasilinear utilities.
As shown in Proposition~\ref{prop:firstWelfare}, our local search gives an analogous bound for Nash social welfare.

We note that local search applied directly to the NSW problem cannot yield a constant factor approximation algorithm even if we allow changing an arbitrary fixed number $k$ of items. This can be seen already when $m=n$, in which case every allocation with positive NSW value is a perfect matching. Also, some other natural variants of local search do not work, or the analysis is not clear; for example, our analysis does not seem to extend when local search is applied to the (seemingly more natural) choice of $\bar{v}_i(S) = v_i(S + \tau(i))$. %

The idea of defining $\ell(i)$ and using the modified valuation functions is inspired by rounding of the fractional solution from previous approaches; the role of the $\ell(i)$ items is similar to the large items in \cite{LiV22}, but we obtain much better guarantees using a more direct deterministic approach. We remark that the iterated continuous greedy algorithm in \cite{LiV22} can be analogously replaced by continuous local search, with a slightly improved approximation factor---this variant is unpublished.

A key part of the analysis is the Rematching Lemma (Lemma~\ref{lem:rematching-estimation}) that shows how the value of the final 
solution $(R_i+\sigma(i))_{i \in A}$ can be related to the 
infeasible allocation $(R_i+\ell(i))_{i \in A}$. 
While the rematching phase was already present (and essentially identical) in \cite{GHV20} and \cite{LiV22}, the analysis here is much simpler. At the heart of the argument is a general monotonicity property of maximum-weight bipartite matchings on subsets of nodes; we discuss this in Appendix~\ref{sec:laci-matching}. 

Finally, we note that local search approaches have been used for objectives related to Nash Social Welfare, albeit in a different context. In \cite{barman2022nash}, local search was used for maximization of a function in the form $\prod_{i \in \aA} (1+f_i(S))$ where $S$ is selected jointly for all agents, subject to certain constraints. In \cite{peters2020proportionality,MavrovMS23}, local search was used for variants of {\em multiwinner elections}, where again one aims to select a shared solution $S$ for all agents (or voters). Objectives such as $\sum_{i\in \aA} \log (1+v_i(S))$, and the related Proportional Approval Voting (PAV) rule, have been used here as a proxy for the goal of finding a ``core'' solution. We note that these algorithms and their analysis are not directly related to ours; in particular because of the baseline value of $+1$ for each agent, a simple local search can be analyzed directly in these models.
 
\section{Finding fair and efficient allocations}\label{sec:fairness}
In this section, we restate and prove  Theorem~\ref{thm:efx2}. 
\efxhalf* 
Our first key tool  is a subroutine that finds a partial allocation that is $\nfrac12$-EFX and preserves a large fraction of the NSW value.

\begin{lemma}\label{lemma:fair-or-efficient}
There exists a deterministic strongly polynomial algorithm \pFairOrEfficient{$\cT$}, that, for any partial allocation $\cT$, returns another partial allocation $\cR$ that satisfies one of the following properties
\begin{enumerate}[label=(\roman*)]
\item\label{i:inc-NSW} $\NSW(\cR)\ge \NSW(\cT)$ and $\cup_{i\in\aA} R_i \subsetneq \cup_{i\in\aA} T_i$, or  
\item\label{i:EFX} $\NSW(\cR)\ge \tfrac12 \NSW(\cT)$ and $\cR$ is $\nfrac12$-EFX.
\end{enumerate}
\end{lemma}
This is shown by  modifying the approach of Caragiannis, Gravin, and Huang~\cite{CaragiannisGH19}. %
The key subroutine for them provides a similar alternative as in Lemma~\ref{lemma:fair-or-efficient}. In outcome \ref{i:EFX}, they have the stronger EFX guarantee, while in outcome \ref{i:inc-NSW}, they show that the NSW value increases by a certain factor. In outcome \ref{i:inc-NSW}, it is not clear how an increase in the NSW value could be shown for subadditive valuations. However, arguing about the support decrease leads to a simpler argument.

In~\cite{CaragiannisGH19}, only a partial EFX allocation is found. Theorem~\ref{thm:efx2} shows the existence of a complete allocation, albeit with the weaker $\nfrac12$-EFX property.
To derive Theorem~\ref{thm:efx2}, we start by repeatedly calling \pFairOrEfficient until outcome \ref{i:EFX} is reached. Note that the outcome~\ref{i:inc-NSW} can only happen at most $m$ times because the number of items in $\cR$ reduces by at least one after each call.

The allocation at this point may be partial. We show that the remaining items can be allocated using the classical \emph{envy-free cycle procedure} by Lipton, Markakis, Mossel, and Saberi \cite{LiptonMMS04}. Even though this procedure is known for the weaker EF1 property~\cite{Budish11}, we show that---after a suitable preprocessing step---it can produce an $\nfrac12$-EFX allocation while not decreasing the NSW value of the allocation. 
\medskip

In the remainder of this section, we prove Theorem~\ref{thm:efx2}. In Section~\ref{sec:fair-or-efficient}, we derive Lemma~\ref{lemma:fair-or-efficient}, and in Section~\ref{sec:compl-partial}, we derive Theorem~\ref{thm:efx2} from this lemma.

\subsection{Finding a fair or an efficient allocation}\label{sec:fair-or-efficient}
The subroutine \pFairOrEfficient{$\cT$}, shown in Algorithm~\ref{alg:efx}, generalizes the algorithm by Caragiannis, Gravin, and Huang~\cite{CaragiannisGH19} from additive to subadditive valuations. We begin by defining the notions of \emph{$\nfrac12$-EFX feasible bundles} and the associated \emph{feasibility graph}.

\begin{definition}[$\nfrac12$-EFX feasible bundles and graph]\label{def:hefx}
Given a partial allocation $\cT=(T_i)_{i\in \aA}$, we say that $T_k$ is a \emph{$\nfrac12$-EFX feasible bundle for agent $i$}, if 
\[v_i(T_k) \ge \tfrac12 \max_{ \ell\in \aA, j\in T_\ell}v_i(T_\ell - j) \,.\]
The \emph{$\nfrac12$-EFX feasibility graph} of $\cT$ is a bipartite graph $\cBP=(A \cup \cT, \aE)$ where the edge set $\aE$ is defined as:
\begin{equation}\label{eqn:E}
\begin{aligned}
\aE = & \left\{ (i, T_i)\ |\ T_i \text{ is } \text{$\nfrac12$-EFX feasible for } i\right\} 
 \ \ \cup  \\
& \left\{(i, T_\ell)\ |\ v_i(T_\ell) > 2v_i(T_i) \text{ and } v_i(T_\ell)\ge \max_{ k\in \aA, j\in T_k}v_i(T_k - j)\right\}\, . 
\end{aligned}
\end{equation}
\end{definition}

The following claim follows directly from the definition.
\begin{claim}\label{cl:factor-2}
In the graph $\cBP=(A \cup \cT, \aE)$, every node $i\in \aA$ has degree at least 1.
\end{claim}

\begin{proof}
    If the bundle $T_i$ is $\nfrac12$-EFX feasible for agent $i$, then by definition, the edge $(i, T_i)$ is included in $\aE$, and thus the degree of $i$ is at least 1. Otherwise, there must exist a $k \in \aA$ and $j\in T_k$ such that $v_i(T_k - j) > 2v_i(T_i)$. Let $\ell\in\aA$ be the agent whose bundle maximizes $v_i(T_{\ell})$ among all agents. Then it follows that $v_i(T_{\ell})\ge v_i(T_k - j) > 2v_i(T_i)$, which implies, by the graph construction, that the edge $(i, T_\ell)\in \aE$. Hence, every node $i\in A$ has degree at least 1.  
\end{proof}

In this section, a \emph{matching} will refer to a matching between agents and bundles (and not between agents and items as in previous sections).  Thus, a matching is a mapping $\rho: \aA\to \cT\cup\{\bot\}$ such that $\rho(i)=\rho(k)$ implies $\rho(i)=\rho(k)=\bot$. A \emph{perfect matching} has $\rho(i)\neq\bot$ for every $i\in\aA$.
Matchings may use pairs $(i,T_k)$ that are not in $\aE$; we say that $\rho$ is a matching in the bipartite graph $\cBP=(\aA\cup\cT,\aE)$ if $(i,\rho(i))\in \aE$ whenever $\rho(i)\neq \{\bot\}$.
 For two matchings $\rho$ and $\tau$, an \emph{alternating path between $\rho$ and $\tau$} is a path
$P = (i_1, S_{i_1}, i_2, \ldots, S_{i_{k-1}}, i_\ell, S_{i_\ell})$ such that $\rho(i_t)=S_{i_t}$, $t=1,\ldots,\ell$, $\tau(i_{t+1})=S_{i_t}$, $t=1,\ldots,\ell-1$. The following lemma is immediate from the definition of the  $\nfrac12$-EFX feasibility graph.
\begin{lemma}\label{lem:pm}
If the $\nfrac12$-EFX feasibility graph $\cBP=(\aA\cup\cT,\aE)$  of an allocation $\cT$ contains a perfect matching $\rho$, then $(i,\rho(i))_{i\in A}$ is a $\nfrac12$-EFX allocation. 
\end{lemma}

\begin{algorithm}[!t]
\DontPrintSemicolon
\KwIn{Partial allocation $\cT$.} 
\KwOut{Partial allocation $\cR$ such that either $\NSW(\cR)\ge \NSW(\cT)$ and $\cup_i R_i \subsetneq \cup_i T_i$, or $\NSW(\cR)\ge \tfrac{1}{2}\NSW(\cT)$ and $\cR$ is $\nfrac12$-EFX.} 
$\cS  \gets \cT$ \;
$\rho(i)\gets \bot$ for all $i\in \aA$\;
\While{$\rho$ is not a perfect matching in $\cBP$}{
$\cBP=(\aA\cup \cS,\aE) \gets \nfrac12$-EFX feasibility graph of $\cS$\tcp*{Definition~\ref{def:hefx}}
$\cL \gets \{S_i, i\in A\ |\ S_i \subsetneq T_i\}$ \tcp*{set of trimmed down bundles}
Define matching $\tau$ with $\tau(i)=S_i$ for all $i\in \aA$\tcp*{candidate matching}
$\rho \gets$ matching in $\cBP$ where\tcp*{Lemma~\ref{lem:T}}\label{line:rho}
\begin{itemize}[topsep=0.3em]
\setlength\itemsep{-0.3em}
\item[$(a)$] all bundles in $\cL$ are matched, 
\item[$(b)$] $|\{i : \rho(i)=S_i\}|$ is maximized subject to $(a)$, and 
\item[$(c)$] $\rho$ is maximum subject to $(a)$ and $(b)$
\end{itemize}

\If{$\exists i_1 \in A$ not matched in $\rho$}{
$(h, g_h) \gets \arg\max_{k\in\aA, g_k\in S_k} v_{i_1}(S_k - g_k)$ \;\label{line:maxb}
\If{$v_h(S_h - g_h) \ge \frac{1}{2} \cdot v_h(T_h)$}
{$S_h\gets S_h-g_h$}
\Else{
$P = (i_1, S_{i_1}, i_2, \dots, S_{i_{\ell-1}}, i_\ell, S_{i_\ell})$ $\gets$ alternating path between $\tau$ and $\rho$ starting at $i_1$ and ending at either $S_{i_\ell} = S_h$ or an unmatched $S_{i_\ell}\neq S_h$\tcp*{Lemma~\ref{lem:P}}\label{line:P}
Construct $\cR$: \; \label{line:Y}
\ \ \ $R_{i_1} \gets S_h - g_h$\; 
\ \ \ \lFor{$f\gets 2$ \KwTo $\ell$}{$R_{i_f} \gets S_{i_{f-1}}$}
\ \ \ \lFor{$i\in \aA\setminus (\{i_1, \dots, i_\ell\} \cup \{h\})$}{$R_i \gets T_i$}
\If{$P$ ends at an unmatched bundle $S_{i_\ell}\neq S_h$}{
\ \ \ $R_h \gets T_h  \setminus (S_h - g_h)$\; 
}
\Return{$\cR$}\label{line:if}\;
}
}
}
\Return{$\cR=(\rho(i))_{i\in A}$}\label{line:last}\;
\caption{{\texttt{MakeFairOrEfficient($\cT$)}}}\label{alg:efx}
\end{algorithm}

We now give an overview of Algorithm~\ref{alg:efx}. For an input partial allocation $\cT=(T_i)_{i\in\aA}$, it returns a partial allocation $\cR$ that satisfies one of the alternatives in Lemma~\ref{lemma:fair-or-efficient}: either \ref{i:inc-NSW} $\NSW(\cR)\ge \NSW(\cT)$ and $\cup_i R_i \subsetneq \cup_i T_i$, or \ref{i:EFX} $\NSW(\cR)\ge \tfrac{1}{2}\NSW(\cT)$ and $\cR$ is  $\nfrac12$-EFX. %

The algorithm gradually \emph{`trims down'} the bundles $\cT$. That is, we maintain a partial allocation $\cS=(S_i)_{i\in \aA}$ with $S_i\subseteq T_i$ throughout; these are initialized as
$S_i=T_i$. %
Every main loop of the algorithm either terminates by constructing an allocation $\cR$ satisfying  \ref{i:EFX}, or removes an item from one of the $S_h$ sets. The other possible termination option is when the $\nfrac12$-EFX feasibility graph of $\cS$ contains a perfect matching $\rho$. In this case, we return $\cR=(\rho(i))_{i\in \aA}$. This is a $\nfrac12$-EFX allocation by  Lemma~\ref{lem:pm}; Lemma~\ref{thm:alg3} shows it also satisfies $\NSW(\cR)\ge \tfrac{1}{2}\NSW(\cT)$ and is thus a suitable output of type \ref{i:EFX}.

At the beginning of each main loop, we define two matchings. The first is the perfect matching $\tau$ that simply defines $\tau(i)=S_i$ for all $i\in\aA$. The second is a matching $\rho$ in $\cBP$. This is required to satisfy three properties: First, it matches all trimmed down bundles, i.e., all bundles $S_i$ with $S_i\subsetneq T_i$. Second, $|\{i: \rho(i) = S_i\}|$ is maximized subject to the first requirement. Third, subject to these requirements, $\rho$ is chosen as a maximal matching. (The existence of such a matching is guaranteed by Lemma~\ref{lem:T} below).

If $\rho$ is not perfect, then we consider an unmatched agent $i_1$, and find the bundle that maximizes $i_1$'s utility after removal of one item. Let $(S_h,g_h)\in \arg\max_{k\in\aA,g\in S_k} v_{i_1}(S_k-g)$.
If agent $h$'s value of $S_h- g_h$ is at least $\frac{1}{2}$ times their value for the original bundle $T_h$, then we remove $g_h$ from $S_h$ and the main loop finishes. Otherwise, we construct an alternating path between $\rho$ and $\tau$, denoted as $P = (i_1, S_{i_1}, i_2, S_{i_2}, \dots, S_{i_{\ell-1}}, i_\ell, S_{i_\ell})$,  starting with $i_1$ and ending with either $S_{i_\ell} = S_h$ or an unmatched bundle $S_{i_\ell} \neq S_h$. Lemma~\ref{lem:P} shows that such a $P$ exists. Using $P$, we construct an allocation $\cR$ in line~\ref{line:Y}. Lemma~\ref{thm:alg3} shows that this is a suitable output of type \ref{i:inc-NSW}.

\subsubsection{Analysis}
The number of iterations of the repeat loop is at most $m$, because in each iteration except the final one, an item is removed from one of the bundles. Noting that a maximum matching in line~\ref{line:rho} and alternating path in line~\ref{line:P} can be found in strongly polynomial-time, Algorithm~\ref{alg:efx} runs in strongly polynomial-time. 

The next lemma guarantees that the matching $\rho$ is well-defined. 
The proof follows similarly as in~\cite{CaragiannisGH19}.

\begin{lemma}\label{lem:T}
In each iteration of the repeat loop in Algorithm~\ref{alg:efx}, a matching exists in $\cBP$ where all bundles of $\cL$ are matched. 
\end{lemma}

\begin{proof}
Let $\aE^{(t)}$ denote the edge set of the $\nfrac12$-EFX feasibility graph  and $\cL^{(t)}$ the set of trimmed down bundles, and $\rho^{(t)}$ the maximum matching in the  $t$-th iteration. 

We show by induction that there exists a matching $\rho^{(t)}$ such that all bundles in $\cL^{(t)}$ are matched. At the beginning of the first iteration, $\cL^{(1)}$ is empty, so the claim is clearly true. Suppose the claim is true until the beginning of $(t+1)$-st iteration.
Let $\cS$ denote the trimmed down bundles in the $t$-th iteration, and let $i_1$ be the unmatched agent, and $(S_h,g_h)$ the bundle and item selected in line~\ref{line:maxb}. 

By the requirement that $|\{i: \rho(i) = S_i\}|$ is maximized subject to all trimmed down bundles being matched,
we have  $(i_1, S_{i_1})\not\in \aE^{(t)}$. 
This means that $2v_{i_1}(S_{i_1})<\max_{\ell\in A, j\in S_{\ell}}v_{i_1}(S_\ell-j)=v_{i_1}(S_h-g_h)$
by the choice of $h$. Thus, $(i_1, S_h) \in \aE^{(t)}$ follows.

Note that $\cL^{(t+1)}= \cL^{(t)}\cup\{h\}$. Consider the $\nfrac12$-EFX feasibility graph in the $(t+1)$-st iteration. Since all bundles different from $S'_h:=S_h-g_h$  remained unchanged, for every edge $(i,S_k)\in \aE^{(t)}$ with $k\neq h$ it follows that $(i,S_k)\in \aE^{(t+1)}$. According to Definition~\ref{def:hefx}, $(i_1, S_h')\in \aE^{(t+1)}$. Let us define $\rho'$ as
\[
\rho'(i):=\begin{cases} S'_h&\mbox{if }i=i_1\, ,\\
\rho^{(t)}(i)&\mbox{if }i\neq i_1\mbox{ and }\rho^{(t)}(i)\neq S_h\, ,\\
\bot&\mbox{otherwise}\, .
\end{cases}
\]
By the above, this gives a matching in $\aE^{(t+1)}$, and it matches all bundles in $\cL^{(t+1)}=\cL^{(t)}\cup\{h\}$.
\end{proof}

\begin{lemma}\label{lem:P}
The alternating path $P$, as described in line~\ref{line:P} of Algorithm~\ref{alg:efx}, exists. 
\end{lemma}

\begin{proof}
Since $i_1$ is an unmatched agent and the requirement that $|\{i: \rho(i) = S_i\}|$ is maximized subject to all trimmed down bundles being matched
in the maximum matching $\rho$ in line~\ref{line:rho}, 
we must have $(i_1, S_{i_1})\not\in \aE$. If $\rho(i_2)= S_{i_1}$ for an agent $i_2\in \aA$,  then we continue with $S_{i_2}$, otherwise we stop. Continuing this way, we eventually reach either $S_{i_\ell} = S_h$ or an unmatched bundle $S_{i_\ell} \neq S_h$.
\end{proof}

\begin{lemma}\label{thm:alg3}
If 
Algorithm~\ref{alg:efx} returns an allocation $\cR$ 
in  line~\ref{line:last}, then $\NSW(\cR)\ge \tfrac{1}{2}\NSW(\cT)$ and $\cR$ is  $\nfrac12$-EFX. If it returns $\cR$ in line~\ref{line:if}, then $\NSW(\cR)\ge \NSW(\cT)$ and $\cup_i R_i \subsetneq \cup_i T_i$.
\end{lemma}

\begin{proof}
Let us start with the case when a perfect matching $\cR=(\rho(i))_{i\in A}$ is returned in  line~\ref{line:last}. The $\nfrac12$-EFX property follows by Lemma~\ref{lem:pm}. Let us show $\NSW(\cR)\ge \tfrac{1}{2}\NSW(\cT)$.

Throughout the algorithm, $v_i(S_i) \ge \frac{1}{2}v_i(T_i)$ is maintained according to the condition on bundle trimming.
By the definition of the $\nfrac12$-EFX feasibility graph,  
either $R_i=S_i$, or $v_i(R_i)>2v_i(S_i)$. Therefore, we have 
\begin{equation}
\begin{aligned}
v_i(R_i) \ge v_i(S_i) \ge \frac{1}{2} v_i(T_i)\, , \quad \forall i\in \aA \, .
\end{aligned}
\end{equation}
Consequently,
\[\NSW(\cR)=\left(\prod_{i\in A} v_i(R_i)\right)^{1/n}  \ge\frac{1}{2}  \left(\prod_{i\in A} v_i(T_i)\right)^{1/n} \ge \tfrac{1}{2}\NSW(\cT)\, .\]

\medskip

Consider now the case when the algorithm terminated with $\cR$ in line~\ref{line:if}. We need to show $\NSW(\cR)\ge \NSW(\cT)$ and $\cup_i R_i \subsetneq \cup_i T_i$. We distinguish two cases based on whether $S_{i_\ell} = S_h$, or whether $S_{i_\ell} (\neq S_h)$ is an unmatched bundle. For the first case, we have 
\begin{equation*}
\begin{aligned}
v_{i_f}(R_{i_f}) & = v_{i_f}(S_{i_{f-1}}) > 2v_{i_f}(S_{i_f}) \ge v_{i_f}(T_{i_f}), \ \forall f\in\{2, \dots, \ell\}\, .
\end{aligned}
\end{equation*}
The first inequality follows from the definition of the $\nfrac12$-EFX feasibility graph, and the second inequality follows from the bundle trimming condition maintained throughout the algorithm. Similarly, we have 
\begin{equation*}
\begin{aligned}
v_{i_1}(R_{i_1}) & = v_{i_1}(S_h - g_h) > 2v_{i_1}(S_{i_1}) \ge v_{i_1}(T_{i_1}) \, .
\end{aligned}
\end{equation*}
This implies 
\[\NSW(\cR)=\left(\prod_{i\in A} v_i(R_i)\right)^{1/n}  >  \left(\prod_{i\in A} v_i(T_i)\right)^{1/n} = \NSW(\cT)\, .\] Since we do not assign $g_h$ to any agent in $\cR$, we must have $\cup_i R_i \subsetneq \cup_i T_i$.

In the second case, since $S_{i_\ell}$ is an unmatched bundle in $\rho$ by the choice of the path $P$, we have $S_{i_\ell}\notin\cL$ by the requirements on $\rho$. That is, $S_{i_\ell}  = T_{i_\ell}$. 
We now make use of subadditivity:
\[ 
v_h(T_h) \le v_h(S_h - g_h) + v_h(T_h \setminus (S_h - g_h))\le \frac{1}{2}v_h(T_h)+v_h(T_h \setminus (S_h - g_h))\, .
\]
Together with the definition of the $\nfrac12$-EFX feasibility graph, this gives
\begin{equation}\label{eqn:Y}
\begin{aligned}
v_{i_\ell}(R_{i_\ell}) & = v_{i_\ell}(S_{i_{\ell-1}}) > 2v_{i_\ell}(S_{i_\ell}) = 2v_{i_\ell}(T_{i_\ell})\, , \\
v_{i_f}(R_{i_f}) & = v_{i_f}(S_{i_{f-1}}) > 2v_{i_f}(S_{i_f}) \ge v_{i_f}(T_{i_f}), \ \forall f\in\{2, \dots, \ell-1\}\, ,\\
v_{i_1}(R_{i_1}) & = v_{i_1}(S_h - g_h) > 2v_{i_1}(S_{i_1}) \ge v_{i_1}(T_{i_1}) \\
v_{h}(R_h) & = v_h(T_h\setminus (S_h - g_h)) > \tfrac{1}{2} v_h(T_h) \, .
\end{aligned}
\end{equation}
By multiplying the above bounds for all agents, we get
\[\NSW(\cR)=\left(\prod_{i\in A} v_i(R_i)\right)^{1/n}  >  \left(\prod_{i\in A} v_i(T_i)\right)^{1/n} = \NSW(\cT)\, .\]
Finally, since we do not assign items in $T_{i_\ell}$ to any agent in $\cR$, we must have $\cup_i R_i \subsetneq \cup_i T_i$. Note that if $T_{i_\ell} = \emptyset$, then $\NSW(\cT) = 0$ and $R_i = \emptyset, \forall i$ is a suitable output of type~\ref{i:EFX}. 
\end{proof}

\subsection{Completing the partial allocation}\label{sec:compl-partial}
In this Section, we derive Theorem~\ref{thm:efx2} from Lemma~\ref{lemma:fair-or-efficient}. The corresponding algorithmic proof is given in Algorithm~\ref{alg:efxm}, which relies on two subroutines \pFairOrEfficient (Algorithm~\ref{alg:efx}) and \pEnvyFreeCycle, %
the latter being an envy cycle procedure adapted from \cite{LiptonMMS04} and described below.

The input of Algorithm~\ref{alg:efxm} is an allocation $\cS$ that is $\alpha$-approximation to the symmetric NSW problem. It then repeatedly calls Algorithm~\ref{alg:efx} until the final allocation is $\nfrac12$-EFX and $2\alpha$-approximation to the symmetric NSW problem. Recall that the output of this subroutine is either a partial allocation $\cT'$ that satisfies either $\NSW(\cT')\ge \NSW(\cT)$ and $\cup_i T_i' \subsetneq \cup_i T_i$, or $\NSW(\cT')\ge \nfrac{1}2\NSW(\cT)$ and $\cT'$ is $\nfrac12$-EFX. Since $\cup_i T_i' \subsetneq \cup_i T_i$ in each call in the first case, the number of calls to Algorithm~\ref{alg:efx} is at most $m$. 

\begin{algorithm}[!t]
\DontPrintSemicolon
\KwIn{Allocation $\cS$ that is $\alpha$-approximation to the NSW problem $(\aA, \aG, (v_i)_{i\in \aA})$.}
\KwOut{Allocation $\cT$ that is $\nfrac12$-EFX and $2\alpha$-approximation to the symmetric NSW problem.}
$\cT\gets \cS$ \;
\While{$\cT$ is not $\nfrac12$-EFX}{
$\cT \gets \pFairOrEfficient(\cT)$\tcp*{Algorithm~\ref{alg:efx}} 
}
$U\gets \aG \setminus \cup_i T_i$\tcp*{set of unallocated items}
\While{$\exists i\in\aA\, , \exists j\in U\, :\, v_i(T_i) < v_i(j)$}{
Let $j\in U$ be such that $v_i(T_i) < v_i(j)$ for some agent $i$\;
$U \gets (U \cup T_i) - j$ \;
$T_i \gets \{j\}$\; 
}
$\cT \gets$ \pEnvyFreeCycle{$\cT,U$} \; 
\Return{$\cT$}\;
\caption{Guaranteeing $\nfrac12$-EFX for the symmetric NSW problem}\label{alg:efxm}
\end{algorithm}

At this point, we have a $\nfrac12$-EFX  partial allocation $\cT$ with  $\NSW(\cT)\ge \tfrac{1}{2}\NSW(\cS)$. The rest of Algorithm~\ref{alg:efxm} allocates the remaining items
$U = \aG\setminus (\cup_{i\in\aA} T_i)$ so that $\NSW(\cT)$ does not decrease, and the $\nfrac12$-EFX property is maintained.

First, we modify the allocation in the second repeat loop to ensure that each agent's value for their bundle is at least their value for each remaining item in $U$. This is done by swapping an agent's bundle $T_i$ with a singleton item $j\in U$ whenever  $i$ values $j$ more than the entire bundle $T_i$.

Finally, we run the envy-cycle procedure \pEnvyFreeCycle{$\cT,U$}, introduced in~\cite{LiptonMMS04}{}, to allocate the remaining unassigned items in $U$, starting from the partial allocation $\cT = (T_i)_{i\in A}$. The procedure maintains and repeatedly updates a directed (envy) graph $D=(A, \aE)$, where each node corresponds to an agent, and there is a directed edge $(i,j)\in \aE$ if agent $i$ envies agent $j$, i.e., $v_i(T_i) < v_i(T_j)$. 

The procedure proceeds as follows. If the graph $D$ contains a directed cycle, then the bundles of the agents along the cycle are \emph{rotated}: each agent receives the bundle of the next agent in the cycle. This rotation strictly improves the utility of each agent involved in the cycle, while leaving the utilities of all other agents unchanged. If the graph contains no cycles, then it must have at least one \emph{source} agent---an agent who is not envied by any other agent. In this case, the procedure selects an arbitrary item from $U$ and assigns it to a source agent. The envy graph is then updated to reflect the new allocation. The procedure continues until all items in $U$ have been assigned. We refer the interested reader to~\cite{LiptonMMS04} for more details. 

\medskip

We now verify the correctness and efficiency of this algorithm. 
\begin{lemma}
The second while loop of Algorithm~\ref{alg:efxm} is repeated at most $nm$ times. It maintains the $\nfrac12$-EFX property and $\NSW(\cT)$ is non-decreasing.
\end{lemma}
\begin{proof}
The bound on the number of swaps follows because every agent $i\in\aA$ may swap their bundle at most $m$ times. After the first swap, they maintain a singleton bundle, and they can swap their bundle for the same item $j$ only once, since their valuation $v_i(T_i)$ strictly increases in each swap.

It is immediate that $\NSW(\cT)$ is non-decreasing. It is left to show that the $\nfrac12$-EFX property is maintained. Let $i\in \aA$ be the agent who swapped their bundle $T_i$ for $T'_i=\{j\}$ in the current iteration. Then, the value of $i$'s own bundle increased while the allocation of everyone else remained the same. Hence, agent $i$ cannot violate the $\nfrac12$-EFX property. For the other agents $k\neq i$, $v_k(T_k)\ge \nfrac12\cdot v_k(T_i'-g)$ for all $g\in T_i'$ trivially holds, since $T_i'$ is a singleton.
\end{proof}
Property \eqref{eq:better-than-singleton} below is satisfied after the second repeat loop. Hence, the next lemma completes the analysis of Algorithm~\ref{alg:efxm}.
\begin{lemma}
The subroutine
\pEnvyFreeCycle{$\cT,U$} terminates in $O(n^3 m)$ time, and $\NSW(\cT)$ is non-decreasing. Assume that $\cT=(T_i)_{i\in\aA}$ is $\nfrac12$-EFX, and 
\begin{equation}\label{eq:better-than-singleton}
v_i(T_i)\ge v_i(j)\quad\forall  i\in \aA, \forall j\in U\, .
\end{equation} 
Then, \pEnvyFreeCycle{$\cT,U$}  also  maintains the $\nfrac12$-EFX property.
\end{lemma}
\begin{proof}
The running time analysis is the same as in \cite{LiptonMMS04}; briefly summarized as follows. Finding and removing a cycle in the envy-graph can be done in $O(n^2)$ time. Further, whenever swapping around a cycle, at least one edge is removed from the envy graph. New edges can only be added when we allocate new items from $U$, with at most $n$ edges every time. Since $|U|\le m$, the total number of new edges added throughout is $nm$. This yields the overall $O(n^3 m)$ bound.

Again, it is immediate that $\NSW(\cT)$ is non-decreasing in every step. We need to show that the $\nfrac12$-EFX property is maintained both when swapping around cycles and when adding new items from $U$. When swapping around a cycle, this follows since the set of bundles remains the same, and no  agent's value decreases. 

Consider the case when a source agent say $i$, gets a new item $j$: their new bundle becomes $T'_i=T_i+j$. Note that $i$ is the only agent whose bundle grows; all other bundles remain the same.
For any $k\neq i$ and any $g\in T'_i$, we can bound
\[
v_k(T'_i-g)=v_k(T_i+j-g)\le v_k(T_i)+v_k(j)\le 2v_k(T_k)\, , 
\]
showing that the $\nfrac12$-EFX property is preserved.
Here, the first inequality follows by subadditivity and monotonicity. The second inequality uses  \eqref{eq:better-than-singleton}, and that  $v_k(T_k)\ge v_k(T_i)$, since $i$ was a source node in the envy graph. 
\end{proof}

\section{Conclusion}\label{sec:conclusion} 
We have shown a $(4+\eps)$-approximation algorithm for the symmetric NSW problem with submodular valuations.
Moreover, our algorithm yields an $\exp\big(\log(3n)+\sum_{i=1}^n w_i\log w_i\big)\cdot (\ee+\eps)$-approximation algorithm for the asymmetric NSW problem under submodular valuations (see Remark~\ref{rem:bento}, \cite{bento}). Subsequent work \cite{bei2025nashsocialwelfaresubmodular,feng2024constant,
fengli} has significantly extended these results, leading to a small constant factor approximation independent of the weights.
However, there are still several directions and open problems to explore. 
An obvious one is to improve the approximation ratio for the symmetric case. The current hardness of approximation stands at $\frac{\ee}{\ee-1}\simeq 1.58$ for submodular valuations, which is the same as the optimal factor for maximizing utilitarian social welfare. It would be interesting to prove a separation between the two optimization objectives for submodular valuations.

There are several open questions on the existence of EFX and its relaxations for submodular valuations. A significant one is: Does a (complete) $\alpha$-EFX allocation exist for $\alpha>1/2$? This is open even without imposing any additional efficiency requirements.

\section*{Acknowledgements}
We are grateful to Tomasz Ponitka for providing an improved version of Lemma~\ref{lem:sym-price-bound} (previous version had factor $6$ instead of $4$) and for pointing out a couple of minor issues in Section~\ref{sec:fairness} of the previous version that have been corrected in this paper. We also thank Bento Natura for the improved bound included in Remark~\ref{rem:bento}.

\bibliographystyle{abbrv}
\bibliography{bibfile} 

\appendix
\section{A general rematching lemma}\label{sec:laci-matching}
We now discuss the general statement on matchings underlying our Rematching Lemma (Lemma~\ref{lem:rematching-estimation}). We first show 
a simple and natural statement on bipartite graphs. We were unable to find a previous explicit occurence in the literature. 
For $0$ and $-\infty$ edge weights, this was shown by Brualdi and Scrigmer~\cite{brualdi1968}, and by Brualdi and Mason~\cite{brualdi1972}. This special case implies that the set system representing a transversal matroid of rank $d$ can be reduced to a representation of size $d$. Lemma~\ref{lem:matching-subset} implies the analogous statement for valuated transversal matroids and has been used implicitly in the literature.

\begin{lemma}\label{lem:matching-subset}
Let $\mathcal{G}=(A,B;c)$ be a complete bipartite graph with edge weights $c\in (\RR\cup\{-\infty\})^{A\times B}$. Let $\tau\,:\,A\to B$ be a maximum-weight $A$-perfect matching, and let $H\coloneqq \tau(A)\subseteq B$ denote the set of matched nodes in $B$. Then, for any $X\subseteq A$, there is a maximum-weight $X$-perfect matching $\rho\,:\,X\to B$ such that $\rho(X)\subseteq H$.
\end{lemma}
\begin{proof}
Take any maximum-weight $X$-perfect matching $\rho\,:\,X\to B$ with minimum size of $\rho(X) \setminus H$. 
For the sake of contradiction, assume $\rho(X) \setminus H\neq \emptyset$ and let $j \in \rho(X) \setminus H$.
Let us consider the symmetric difference $\tau\Delta \rho$; this consists of cycles and paths. Since $j \not \in H$, $j$ is an endpoint of an alternating path $P$ between $\rho$ and $\tau$ that starts in $j$ with a $\rho$-edge and ends in $H$ with a $\tau$-edge (since $\tau$ is an $A$-perfect matching).  
Let $Y\subseteq A$ be the set of nodes incident to $P$. Clearly, $Y\subseteq X$. We get another $A$-perfect matching $\tau'\,:\, A\to B$ by flipping $\tau$ along $P$, i.e., setting $\tau'(i)=\rho(i)$ for $i\in Y$ and $\tau'(i)=\tau(i)$ for $i\in A\setminus Y$. Similarly, we obtain another $X$-perfect matching $\rho'\,:\, X\to B$ by flipping $\rho$ along $P$, i.e., setting $\rho'(i)=\tau(i)$ for $i\in Y$ and $\rho'(i)=\rho(i)$ for $i\in X\setminus Y$.

By the optimality of $\tau$, the weight of $\tau$-edges in $P$ is at least the weight of $\rho$-edges in $P$. This implies that the weight of $\rho'$ is at least the weight of $\rho$. This contradicts the choice of $\rho$ since $\rho'(X)\setminus H= (\rho(X)\setminus H)\setminus\{j\}$.
\end{proof}

Lemma~\ref{lem:rematching-estimation} requires a slight extension of the above statement, namely when we may create multiple copies of the nodes in $B\setminus H$. The proof is similar as above.
\begin{lemma}\label{lem:matching-copy}
Let $\mathcal{G}=(A,B;c)$ be a complete bipartite graph with edge weights $c\in (\RR\cup\{-\infty\})^{A\times B}$. Let $\tau\,:\,A\to B$ be a maximum-weight $A$-perfect matching, and let $H\coloneqq \tau(A)\subseteq B$ denote the set of matched nodes in $B$. For some integer $k\ge 1$, let us obtain from $\mathcal{G}$ the bipartite graph $\mathcal{G}'=(A,B';c')$ by creating $k$ copies of every node in $B\setminus H$; thus, $|B'|=k|B|-(k-1)|A|$. Then, for any $X\subseteq A$, there is a maximum-weight $X$-perfect matching $\rho\,:\,X\to B'$ such that $\rho(X)\subseteq H$.
\end{lemma}
\begin{proof}
The case $k=1$ is identical to Lemma~\ref{lem:matching-subset}.
The proof of $k>1$ follows similarly: we consider the symmetric difference of $\tau$ (naturally embedded into the larger graph $\mathcal{G}'$) and $\rho$. We take an alternating path $P$ starting from a node  $j'\in \rho(X)\setminus H$. Assume $j'$ is one of the $k$ copies of a node $j\in B\setminus H$.
 The $A$-perfect matching $\tau'\,:\,A\to B'$ obtained from $\tau$ by flipping along $P$ will assign $\tau'(i')=j'$ for some $i'\in A$. This naturally maps back to an $A$-perfect matching $\bar\tau\,:\,A\to B$ of the same weight in the original graph, by setting $\bar\tau(i')=j$, and $\bar\tau(k)=\tau'(k)$ for $k\neq i'$. This again shows that the weight of $\tau$-edges in $P$ is at least the weight of $\rho$-edges in $P$, and we obtain a contradiction by constructing a better $X$-perfect matching $\rho'\,:\,A\to B'$ as in the $k=1$ case.
\end{proof}

\end{document}